\documentclass[a4paper,UKenglish,cleveref, autoref, thm-restate]{lipics-v2021}
\title{Reducing the Vertex Cover Number via Edge Contractions} %TODO Please add

\titlerunning{Reducing the Vertex Cover Number via Edge Contractions} %TODO optional, please use if title is longer than one line

\author{Paloma T. Lima}{Computer Science Department, IT University of Copenhagen, Denmark}{palt@itu.dk}{}{}%TODO mandatory, please use full name; only 1 author per \author macro; first two parameters are mandatory, other parameters can be empty. Please provide at least the name of the affiliation and the country. The full address is optional

\author{Vinicius F.\ dos Santos}{Departamento de Ci\^encia da Computa\c{c}\~{a}o, Universidade Federal de Minas Gerais, Belo Horizonte, Brazil} {viniciussantos@dcc.ufmg.br}{https://orcid.org/0000-0002-4608-4559}{Grant APQ-01707-21 Minas Gerais Research Support Foundation (FAPEMIG) and Grants 311679/2018-8 National Council for Scientific and Technological Development (CNPq).}

\author{Ignasi Sau}{LIRMM, Universit\'e de Montpellier, CNRS, Montpellier, France}{ignasi.sau@lirmm.fr}{https://orcid.org/0000-0002-8981-9287}{CAPES-PRINT Institutional Internationalization Program, process 88887.371209/ 2019-00, and projects DEMOGRAPH (ANR-16-CE40-0028), ESIGMA (ANR-17-CE23-0010), ELIT (ANR-20-CE48-0008-01), and UTMA (ANR-20-CE92-0027).}

\author{U\'everton S. Souza}{Instituto de Computa\c c\~ao, Universidade Federal Fluminense, Niter\'oi, Brazil \and Institute of Informatics, University of Warsaw, Warsaw, Poland}{ueverton@ic.uff.br}{https://orcid.org/0000-0002-5320-9209}{
This research has received funding from Rio de Janeiro Research Support Foundation (FAPERJ) under grant agreement E-26/201.344/2021,  National Council for Scientific  and Technological Development (CNPq) under grant agreement 309832/2020-9,  and the European Research Council (ERC) under the European Union's Horizon $2020$ research and innovation programme under grant agreement CUTACOMBS (No. $714704$). }

\author{Prafullkumar Tale}{CISPA Helmholtz Center for Information Security, SIC, Saarbr$\ddot{\text{u}}$cken, Germany}{prafullkumar.tale@cispa.saarland}{}{This research is a part of a project that has received funding from the European Research Council (ERC) under the European Union's Horizon $2020$ research and innovation programme under grant agreement SYSTEMATICGRAPH (No. $725978$).\flag{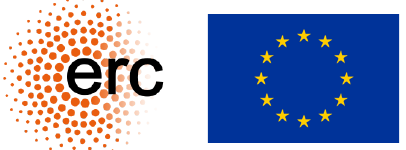}}
\authorrunning{P. T. Lima, V. F. dos Santos, I. Sau, U. S. Souza, and P. Tale} %TODO mandatory. First: Use abbreviated first/middle names. Second (only in severe cases): Use first author plus 'et al.'

\Copyright{Paloma T. Lima, Vinicius F.\ {dos Santos}, Ignasi Sau, U\'everton S. Souza, and Prafullkumar Tale} %TODO mandatory, please use full first names. LIPIcs license is "CC-BY";  http://creativecommons.org/licenses/by/3.0/


% W-hardness.} %TODO mandatory; please add comma-separated list of keywords

\category{} %optional, e.g. invited paper

\relatedversion{} %optional, e.g. full version hosted on arXiv, HAL, or other respository/website
\nolinenumbers %uncomment to disable line numbering

\hideLIPIcs  %uncomment to remove references to LIPIcs series (logo, DOI, ...), e.g. when preparing a pre-final version to be uploaded to arXiv or another public repository

\EventEditors{}
\EventNoEds{0}
\EventLongTitle{}
\EventShortTitle{}
\EventAcronym{}
\EventYear{}
\EventDate{}
\EventLocation{}
\EventLogo{}
\SeriesVolume{}
\ArticleNo{}
\usepackage[svgnames]{xcolor}
\usepackage{color}

\definecolor{darkmagenta}{rgb}{0.30, 0.0, 0.30}

\definecolor{dark-red}{rgb}{0.7,0.15,0.15}
\definecolor{dark-blue}{rgb}{0.15,0.15,0.4}
\definecolor{medium-blue}{rgb}{0,0,0.5}
\definecolor{gray}{rgb}{0.5,0.5,0.5}
\definecolor{color-Ig}{rgb}{0.15,0.7,0.15}

\hypersetup{
    colorlinks, linkcolor={blue},
    citecolor={dark-red}, urlcolor={medium-blue}
}

\newenvironment{proofThm2}[1][]{\medskip \par \noindent {\bf \color{lipicsGray}{Proof of \autoref{thm:w1-hard}.}#1}\ }{\hfill\color{lipicsGray}{$\blacktriangleleft$}\medskip}

\usepackage{amssymb}
\usepackage{amsmath}
\usepackage{amsthm}
\usepackage[sort]{cite}
\usepackage{tikz}
\usetikzlibrary{shapes,arrows}

\usepackage[linesnumbered,algoruled,boxed,lined]{algorithm2e}
\usepackage{complexity} % For symbols in complexity class.
\usepackage{todonotes} % For todo notes
\presetkeys{todonotes}{inline}{}
\newcommand{\cost}{\textsf{cost}}
\newcommand{\el}{\textsf{el}}

\newcommand{\bc}{\texttt{bc}}
\newcommand{\oct}{\texttt{oct}}
\newcommand{\vc}{\texttt{vc}}

\newcommand{\calA}{\mathcal{A}}

\newcommand{\calB}{\mathcal{B}}
\newcommand{\calC}{\mathcal{C}}
\newcommand{\calD}{\mathcal{D}}
\newcommand{\calE}{\mathcal{E}}

\newcommand{\calF}{\mathcal{F}}
\newcommand{\calG}{\mathcal{G}}

\newcommand{\calI}{{\mathcal I}}

\newcommand{\calM}{\mathcal{M}}
\newcommand{\calO}{\ensuremath{{\mathcal O}}}
\newcommand{\OO}{\mathcal{O}}

\newcommand{\calP}{\mathcal{P}}

\newcommand{\rank}{{\sf rank}}

\newcommand{\calT}{\mathcal{T}}

\newcommand{\calW}{\mathcal{W}}
\newcommand{\calY}{\mathcal{Y}}

\newcommand{\ETH}{\textsf{ETH}\xspace}

\newcommand{\Hard}{\textsf{hard}}
\newcommand{\Hardness}{\textsf{hardness}}
\newcommand{\para}{\textsf{para}}

\newcommand{\yes}{{\sc Yes}}
\newcommand{\no}{{\sc No}}

\newtheorem{reduction rule}{Reduction Rule}[section]
\newtheorem{marking-scheme}{Marking Scheme}[section]
\newtheorem{guarantee}{Guarantee}[section]

\newcommand{\defproblem}[3]{
  \vspace{1mm}
\noindent\fbox{
  \begin{minipage}{0.96\textwidth}
  \begin{tabular*}{\textwidth}{@{\extracolsep{\fill}}lr} #1 \\ \end{tabular*}
  {\bf{Input:}} #2  \\
  {\bf{Question:}} #3
  \end{minipage}
  }
  \vspace{1mm}
}

\begin{document}

\maketitle

\begin{abstract}
The \textsc{Contraction(\textsf{vc})} problem takes as input a graph $G$ on $n$ vertices and two integers $k$ and $d$, and asks whether one can contract at most $k$ edges to reduce the size of a minimum vertex cover of $G$ by at least $d$.
Recently, Lima et al.~[JCSS 2021] proved, among other results, that unlike most of the so-called \emph{blocker problems}, \textsc{Contraction(\textsf{vc})} admits an \XP\ algorithm running in time $f(d) \cdot n^{\calO(d)}$. They left open the question of whether this problem is  \FPT\ under this parameterization.
In this article, we continue this line of research and prove the following results:
\begin{itemize}
\item \textsc{Contraction(\textsc{\vc})} is \W[1]-\Hard\ parameterized by $k + d$.
Moreover, unless the \ETH\ fails, the problem does not admit an algorithm running in time $f(k + d) \cdot n^{o(k + d)}$ for any function $f$. In particular, this answers the open question stated in Lima et al.~[JCSS 2021] in the negative.
\item It is \NP-\Hard\ to decide whether an instance $(G, k, d)$ of \textsc{Contraction(\textsf{\vc})} is a \yes-instance even when $k = d$, hence enhancing our understanding of the classical complexity of the problem.
\item \textsc{\textsc{Contraction(\textsc{\vc})}} can be solved in time $2^{\calO(d)} \cdot n^{k - d + \calO(1)}$. This \XP\ algorithm improves the one of Lima et al.~[JCSS 2021], which uses Courcelle's theorem as a subroutine and hence, the $f(d)$-factor in the running time is non-explicit and probably very large. On the other hard, it shows that when $k=d$, the problem is \FPT\ parameterized by $d$ (or by $k$).%can be solved in \FPT-time.
    %\igm{I still think we should calculate the exact exponent}
    %\igm{A reviewer says that it should be $n^{\min\{k - d, d\}}$. In fact, I don't see why we need $\min\{k - d, d\}$. In the first two ``easy'' cases of the algorithm (see \autoref{fig:diagram-algo}, we have a running time $2^{\calO(d)} \cdot n^{\calO(1)}$, but not $n^d$. I think we should say that the algorithm runs in time $2^{\calO(d)} \cdot n^{k-d + \calO(1)}$. Well, we should write explicitly what this $\calO(1)$ is the exponent is, if we have time. Most of the polynomial-time routines that we apply are linear or quadratic.}
\end{itemize}
\end{abstract}

%\newpage
%\setcounter{page}{1}

\section{Introduction}
\label{sec:intro}

%\igm{I have rephrased and added several sentences here and there. I have added a comment only for things where we may need some discussion, or that I consider somehow important}

Graph modification problems have been extensively studied in theoretical computer science, in particular for their vast expressive power and their applicability in a number of scenarios.
Such problems can be generically defined as follows.
For a fixed graph class ${\cal F}$ and a fixed set ${\calM}$ of allowed graph modification operations, such as vertex deletion, edge deletion, edge addition, edge editing or edge contraction, the $\calF$-$\calM$-\textsc{Modification}  problem takes as input a graph $G$ and a positive integer $k$, and the goal is to decide whether at most $k$ operations from ${\cal M}$ can be applied to $G$ so that the resulting graph belongs to the class ${\cal F}$.
For most natural graph classes $\calF$ and modification operations ${\cal M}$, the $\calF$-$\calM$-\textsc{Modification} problem is \NP-hard~\cite{10.1145/800133.804355,WatanabeAN83,LewisY80}.
To cope up with this hardness, these problems have been examined via the lens of parameterized complexity~\cite{BodlaenderHL14grap,crespelle2020survey}.
%NeedFix \igm{use the ``cite'' package, so that references are ordered}.
With an appropriate choice of $\calF$ and the allowed modification operations $\calM$,  \textsc{$\calF$-$\calM$-Modification} can encapsulate well-studied problems like \textsc{Vertex Cover}, \textsc{Feedback Vertex Set (FVS)}, \textsc{Odd Cycle Transversal (OCT)},  \textsc{Chordal Completion},  or \textsc{Cluster Editing}, to name just a few.

The most natural and well-studied modification operations are, probably in this order, vertex deletion,  edge deletion,  and edge addition.
In recent years, the \emph{edge contraction} operation has begun to attract significant scientific attention.
%This operation, which is arguably the most natural edit operation apart from ones mentioned above, is
(When contracting an edge $uv$ in a graph $G$, we delete $u$ and $v$ from $G$, add a new vertex and make it adjacent to vertices that were adjacent to $u$ or $v$.)
In parameterized complexity, $\calF$-\textsc{Contraction} problems, i.e., $\calF$-$\calM$-\textsc{Modification} problems where the only modification operation in $\calM$ is  edge contraction,  are usually studied with the number of edges allowed to contract, $k$, as the parameter.
%For instance, Heggernes et al.~\cite{heggernes2014contracting} proved that if $\calF$ is the set of acyclic graphs then $\calF$-\textsc{Contraction} is \FPT.
A series of recent papers studied the parameterized complexity of $\calF$-\textsc{Contraction} for various graph classes $\calF$ such as paths and trees~\cite{heggernes2014contracting},
generalizations and restrictions of trees~\cite{agarwal2019parameterized, agrawal2017paths},
cactus graphs~\cite{krithika2018fpt},
bipartite graphs~\cite{guillemot2013faster, heggernes2013obtaining},
planar graphs~\cite{golovach2013obtaining},
grids~\cite{saurabh2020parameterized},
cliques~\cite{cai2013contracting},
split graphs~\cite{agrawal2019split},
chordal graphs~\cite{lokshtanov2013hardness},
bicliques~\cite{martin2015computational}, or
degree-constrained graph classes~\cite{Belmonte:2014,golovach2013increasing, DBLP:conf/iwpec/0001T20}.%\igm{I think each of these references is relevant! I would keep all of them. At least for the journal version}

For all the $\calF$-$\calM$-\textsc{Modification} problems mentioned above,  a typical definition of the problem contains a description of the target graph class $\calF$.
For example,  \textsc{Vertex Cover}, \textsc{FVS},  and \textsc{OCT} are $\calF$-$\calM$-\textsc{Modification} problems where $\calF$ is the collection of edgeless graphs, forests, and bipartite graphs, respectively, and $\calM$ contains only vertex deletion.
Recently, a different formulation of these graph modification problems, called \emph{blocker problems},  has been considered.
In this formulation, the target graph class is defined in a \emph{parametric way} from the input graph.
To make the statement of such problems precise, consider an invariant~$\pi : \calG \mapsto \mathbb{N}$, where $\calG$ is the collection of all graphs.
For a fixed invariant $\pi$,  a typical input of a blocker problem consists of a graph $G$, a budget $k$,  and a threshold value $d$,  and the question is whether $G$ can be converted into a graph $G'$ using at most $k$ allowed modifications such that $\pi(G') \le \pi(G) - d$.
This is the same as determining whether $(G, k,d)$ is a \yes-instance of $\calF^{\pi}_{G, d}$-$\calM$-\textsc{Modification} where $\calF^{\pi}_{G, d} = \{G' \in \calG \mid \pi(G') \le \pi(G) - d\}$.
%Note that, here the target graph class is defined in a parametric way from an input graph.

Consider the following examples of this formulation.
For  invariant $\pi(G) = |E(G)|$,  threshold $d = |E(G)|$,  and vertex deletion as the modification operation in $\calM$, $\calF^{\pi}_{G, d}$-$\calM$-\textsc{Modification} is the same as \textsc{Vertex Cover}.
Setting the threshold $d$ to a fixed integer $p$ leads to \textsc{Partial Vertex Cover}.
In a typical definition of this problem,  the input is a graph $G$ and two integers $k, p$,  and the objective is to decide whether there is a set of vertices of size at most $k$ that has at least $p$ edges  incident on it.
Consider another example when $\pi(G) = \vc(G)$,  the size of a minimum vertex cover of $G$,  the threshold value $d = \vc(G) - 1$,   and the allowed modification operation is edge contraction.
To reduce the size of a minimum vertex cover from $\vc(G)$ to $1$ by $k$ edge contractions, we need to find a connected vertex cover of size $k + 1$.
Hence, in this case $\calF^{\pi}_{G, d}$-$\calM$-\textsc{Modification} is the  same as \textsc{Connected Vertex Cover}.
In all these cases,  we can think of the set of vertices or edges involved in the modifications as `blocking'  the invariant $\pi$, that is, preventing $\pi$ from being smaller.

Blocker problems have been investigated for numerous graph invariants, such as the chromatic number, the independence number, the matching number, the diameter, the domination number, the total domination number, and the clique number of a graph~\cite{BTT11, CWP11, bazgan2015blockers, diner2018contraction, kimarxiv, PBP, paulusma2018critical,RBPDCZ10,WatanabeAN83} with `vertex deletion' or `edge deletion' as the allowed graph modification operation.
Blocker problems with the edge contraction operation have already been studied with respect to the chromatic number, clique number, independence number~\cite{diner2018contraction,paulusma2018critical}, the domination number~\cite{GALBY2021DM,GalbyArxivP3}, total domination number~\cite{GalbyArxivTotalD}, and the semitotal domination number~\cite{galby2021using}.

This article is strongly motivated by the results in \cite{lima2020reducing}.
They proved,  among other results,  that it is  \coNP-\Hard\ to test whether we can reduce the size of a minimum feedback vertex set or of  a minimum odd cycle transversal of a graph by one, i.e., $d = 1$, by performing one edge contraction, i.e., $k = 1$.
This is consistent with earlier results, as blocker problems are generally very hard, and become polynomial-time solvable only when restricted to specific graph classes.
However,  the notable exception is the case when the invariant in question is the size of a minimum vertex cover of the input graph.
We  formally define the  problem before mentioning existing results and our contribution (where $G/F$ denotes the graph obtained from $G$ by contracting the edges in $F$).

\medskip
\defproblem{\textsc{Contraction(\vc)}}{An undirected graph $G$ and two non-negative integers $k$ and  $d$.}{Does there exist a set $F \subseteq E(G)$ such that $|F| \le k$ and $\vc(G/F) \le \vc(G) - d$?}

\medskip
\noindent \textbf{Our results}.
A simple reduction,  briefly mentioned in \cite{lima2020reducing},  shows that the above problem is \NP-\Hard\ for \emph{some} $k$ in $\{d, d+1, \dots, 2d\}$.
In our first result,  we enhance our understanding of the classical complexity of the problem and prove that the problem is \NP-\Hard\ even when $k = d$.
As any edge contraction can decrease $\vc(G)$ by at most one, if $k < d$ then the input instance is a trivial \no-instance.
To state our first result,  we introduce the notation of $\rank(G)$, which  is the number of vertices of $G$ minus its number of connected components (or equivalent, the number of edges of a set of spanning trees of each of the connected components of $G$).
Note that it is sufficient to consider values of $k$ that are at most $\rank(G)$, as otherwise it is possible to transform $G$ to an edgeless graph with at most $k$ contractions.
\begin{theorem}
\label{thm:np-hard}
To decide whether an instance $(G, k,  d)$ of {\sc Contraction(\vc)} is a {\sc Yes}-instance is
\begin{itemize}
\item \emph{\coNP-\Hard} if $k = \emph{\rank}(G)$,
\item \emph{\coNP-\Hard} if $k < \emph{\rank}(G)$ and $2d \le k$,  and
\item \emph{\NP-\Hard} if $k < \emph{\rank}(G)$ and $k = d + \frac{\ell - 1}{\ell + 3} \cdot d$ for  any  integer $\ell \ge 1$ such that $k$ is an integer.
\end{itemize}
\end{theorem}

As one needs to contract at least $d$ edges to reduce the size of a minimum vertex cover by~$d$, the above theorem, for $\ell = 1$,  implies that the problem is \para-\NP-\Hard\ when parameterized by the `excess over the lower bound', i.e., by $k - d$.
Since we can assume that $d \le k$, $d$~is a `stronger' parameter than $k$.
One of the main results of~\cite{lima2020reducing} is an \XP\ algorithm for \textsc{Contraction(\vc)} with running time $f(d) \cdot n^{\calO(d)}$.
Here, and in the rest of the article, we denote by $n$ the number of vertices of the input graph.
The authors explicitly asked whether the problem admits an \FPT\ algorithm parameterized by $d$.
As our next result, we answer this question in the negative by proving that such an algorithm is highly unlikely, even when parameterized by the larger parameter $d + k$ (or equivalently, just $k$, as discussed above).

%\igm{We have said, a few lines above, that we can assume that $k \geq d$, so taking $k+d$ as the parameter is the same as taking just $k$. Anyway, it may be a good idea to leave both, in case some reader forgets that $k \geq d$}

\begin{theorem}
\label{thm:w1-hard}
{\sc Contraction(\vc)} is \emph{\W[1]-\Hard} parameterized by $k+d$.
Moreover, unless the \emph{\ETH} fails\footnote{The Exponential Time Hypothesis (\ETH) is a conjecture stating that, roughly speaking, $N$-variable $3$-SAT cannot be solved in time $2^{o(N)}$.
We refer readers to \cite[Chapter $14$]{cygan2015parameterized} for the formal definition and the related literature.}, it does not admit an algorithm running in time $f(k + d) \cdot n^{o(k + d)}$ for any computable function $f: \mathbb{N} \mapsto \mathbb{N}$. The result holds even if we assume that the input $(G,k,d)$ is such that $k < \rank(G)$ and $d \le k < 2d$, and $G$ is a bipartite graph with a bipartition $\langle X, Y \rangle$ such that $X$ is a minimum vertex cover of $G$.
\end{theorem}

For the \XP\ algorithm in \cite{lima2020reducing}, the authors did not explicitly mention an upper bound on the corresponding function $f$, but it is expected to be quite large since the algorithm uses Courcelle's theorem~\cite{Courcelle90} as a subroutine.
Our next result improves this algorithm by providing a concrete upper bound on the running time, and by distinguishing in a precise way the contribution of $k$ and $d$.
\begin{theorem}
\label{thm:algorithm}
There exists an algorithm that solves {\sc Contraction(\vc)} in time $2^{\calO(d)} \cdot n^{k-d + \calO(1)}$. Moreover, for an input $(G,k,d)$, the algorithm runs in time $2^{\calO(d)} \cdot n^{\calO(1)}$ unless $k <  {\sf rank}(G)$ and $d \leq k < 2d$.
\end{theorem}
Note that the above result implies, in particular, that the problem is \FPT\ parameterized by $d$ when $k - d$ is a constant.

\medskip
\noindent \textbf{Our methods}. A central tool in both our negative and positive results is \autoref{lemma:solution-edge-pair}, which allows us to reformulate the problem as follows. As discussed later, by applying appropriate \FPT\ reductions to the input graph $G$, it is possible to assume that we have at hand a minimum vertex cover $X$ of $G$. We say that a set of edges $F$ is a \emph{solution} of $(G, k, d)$ if $|F| \le k$ and $\vc(G/F) \le \vc(G) - d$.
\autoref{lemma:solution-edge-pair} implies that there exists such a solution (i.e., an edge set) if and only if there exist vertex subsets $X_s \subseteq X$ and $Y_s \subseteq V(G) \setminus X$ such that the pair $\langle X_s, Y_s \rangle$ satisfies the technical conditions mentioned in its statement (and which we prefer to omit here).
This reformulation allows us to convert the problem of finding a subset $F$ of edges to the problem of modifying the given minimum vertex cover $X$ to obtain another vertex cover $X_{\el} = (X \setminus X_s) \cup Y_s$ such that $|X_{\el}| \le |X| + (k - d)$ and $\rank(X_{\el}) \ge k$.
Here, we define $\rank(X_{\el}) := \rank(G[X_{\el}])$.
See \autoref{fig:sol-edge-sol-pair} for an illustration.

\smallskip

In our hardness reductions, another simple, yet critical, tool is  \autoref{lemma:pendant-out-of-pair}, which states that if there is a vertex which is adjacent to a pendant vertex, then there is a solution pair that does not contain this vertex.
We present overviews of the reductions in \autoref{sec:np-hardness} and \autoref{sec:w-hard} to  demonstrate the usefulness of these two lemmas in the respective hardness results. The reduction that we use to prove the third item in the statement of \autoref{thm:np-hard} (which is the most interesting case) is from a variant of \textsc{Multicolored Independent Set}, while the one in the proof of \autoref{thm:w1-hard} is from \textsc{Edge Induced Forest}, a problem that we define and that we prove to be \W[1]-\Hard\ in \autoref{thm:w-hardness-edge-induced-forest}, by a parameter preserving reduction from, again, \textsc{Multicolored Independent Set}. It is worth mentioning that the \W[1]-\Hardness\ in \autoref{thm:w-hardness-edge-induced-forest} holds even if we assume that the input graph $G$ is a bipartite graph with a bipartition $\langle X, Y \rangle$ such that $X$ is a minimum vertex cover of $G$, and such that $k < \rank(G)$ and $d \le k < 2d$. This case is the crux of the difficulty of the problem. %as it
This becomes clear in the \XP\ algorithm of \autoref{thm:algorithm} that we proceed to discuss.

%As discussed in \autoref{sec:algorithm}, it is easy to see that the crux of the difficulty of the problem lies in the case when $k < \rank(G)$ and $d \le k < 2d$.
%\red{We can even assume that $G$ is a bipartite graph with a bipartition $\langle X, Y \rangle$ such that $X$ is a minimum vertex cover of $G$, and the problem remains equally hard.}\igm{Why can we assume it? I find no explicit statement of this claim in the whole paper. And I don't see how this is exploited in \autoref{thm:np-hard} and \autoref{thm:w1-hard} (see my comments there). See also my comment on page~\pageref{comment:assumption-met}, and the page before}
%We exploit this difficult case to prove \autoref{thm:np-hard} and \autoref{thm:w1-hard}.
%In order to do that, we need \autoref{lemma:solution-edge-pair} which is a central tool in all of our results.

\smallskip

The algorithm for \textsc{Contraction(\vc)}, which is our main technical contribution, is provided in \autoref{sec:algorithm}. A diagram of this algorithm is shown in \autoref{fig:diagram-algo}. By a standard Knapsack-type dynamic programming table, which is also mentioned in \cite{lima2020reducing}, we can assume that the input graph $G$ is connected. We distinguish three cases depending on the relation between $k, d$, and $\rank(G)$. The first two cases are easy, and can be solved in time $2^{\calO(d)} \cdot n^{\calO(1)}$,  by essentially running an \FPT\ algorithm to determine whether $\vc(G) < d$; see
\autoref{lemma:k-equal-rank-algo} and \autoref{lemma:k-larger-2k-algo}.
We now present  an overview of the algorithm for the third case, namely when its input $(G, k, d)$ is with guarantees that $k < \rank(G)$ and $d \le k < 2d$ (cf. \autoref{lemma:k-lesser-2k-algo}).
Inspired by \autoref{lemma:solution-edge-pair},  we introduce an annotated version of the problem called \textsc{Annotated Contraction(\vc)}.
We first argue  (cf. \autoref{lemma:large-oct-large-rank-vc}) that there is an algorithm that runs in time $2^{\calO(k)} \cdot n^{\calO(1)}$, and either correctly concludes that $(G, k, d)$ is a \yes-instance of \textsc{Contraction(\vc)} or finds a minimum vertex cover $X$ of $G$ such that $\rank(X) < d$.
Using this vertex cover, we can construct $2^{\calO(d)}$ many instances of \textsc{Annotated Contraction(\vc)} such that $(G, k, d)$ is a \yes-instance of \textsc{Contraction(\vc)} if and only if at least one of these newly created instances is a \yes-instance of \textsc{Annotated Contraction(\vc)} (cf. \autoref{lemma:contr-VC-to-annot-contr-vc}).
Hence, it suffices to design an algorithm to solve \textsc{Annotated Contraction(\vc)}. We show that we can apply a simple reduction rule (cf. \autoref{lemma:rr-anno-contr-edges-X}) that allows us to assume that the input graph $G$ of
\textsc{Annotated Contraction(\vc)} is
 bipartite with bipartition $\langle X, Y \rangle$ such that $X$ is a minimum vertex cover of $G$, as mentioned above.

\begin{figure}[t]
  \begin{center}
\tikzstyle{input} = [rectangle, draw, fill=green!80, text width=8em, text centered, rounded corners, minimum height=4em]
\tikzstyle{algo} = [rectangle, draw, fill=red!20, text width=8em, text centered, rounded corners, minimum height=4em]
\tikzstyle{reduction} = [rectangle, draw, fill=blue!20, text width=8em, text centered, rounded corners, minimum height=4em]
\tikzstyle{decision} = [diamond, draw, fill=cyan!20, text width=4.5em, text badly centered, node distance=3cm, inner sep=0pt]
\tikzstyle{line} = [draw, -latex']

\scalebox{0.85}{
\begin{tikzpicture}[node distance = 2cm, auto]
% Place nodes
\node [input] (inputnode) {Instance $(G, k, d)$ of \textsc{Contr.(\vc)}};
\node [algo, below of=inputnode, node distance=2cm] (keqrankG) {Solve in $\calO^{\star}(2^{\calO(d)})$ using \autoref{lemma:k-equal-rank-algo}};
\node [algo, right of=keqrankG, node distance=5cm] (klerankG2dlek) {Solve in $\calO^{\star}(2^{\calO(d) })$ using \autoref{lemma:k-larger-2k-algo}};
\node[diamond, draw, fill=cyan!20, right of=inputnode, node distance=5cm] (invisiblenode) {};
\node [decision, right of=inputnode, node distance=10cm] (lemma-oct) {\autoref{lemma:large-oct-large-rank-vc} $\calO^{\star}(2^{\calO(k)})$};
\node [algo, above of=invisiblenode, node distance=1.5cm] (yesinstance) {$(G, k, d)$ is a \yes-instance};
\node [reduction, below of=lemma-oct, node distance=2.5cm]  (vc-x-rank-d) {Min. vertex cover $X$ of $G$ with $\rank(X) < d$};
\node [reduction, below of=vc-x-rank-d, node distance=2cm] (anno-contr-vc) {$2^{\calO(d)}$ instances of \textsc{Annotated Contr.(\vc)}};
\node [reduction, below of=klerankG2dlek, node distance=2.5cm] (input-bipartite) {$G$ is bipartite with bipartition $\langle X, Y \rangle$ and $X$ is min. vertex cover};
\node [reduction, below of=keqrankG, node distance=2.5cm] (constrained-max-cut) {Instance $((G, k, d), \cdot)$ of \textsc{Constrained MaxCut}};
\node [diamond, draw, fill=cyan!20, below of=constrained-max-cut, node distance=2.5cm] (is-k-eq-d) {};
\node [reduction, below of=is-k-eq-d, node distance=2cm]
(xp-factor) {Create $2^{\calO(d)} \cdot n^{k - d}$ many instances such that $k = d$ using \autoref{lemma:k-neq-d-to-k-eq-d}};
%\node [reduction, below of=is-k-eq-d, node distance=2cm]
%(xp-factor) {Create $2^{\calO(d)} \cdot n^{k - d}$ instances with prop. $k = d$ using \autoref{lemma:k-neq-d-to-k-eq-d}};
\node [reduction, below of=input-bipartite, node distance=2.5cm] (matching-x-y) {Simplify using \autoref{lemma:rr-only-matching} to get $|X| = |Y|$};
%\node [reduction, below of=input-bipartite, node distance=2.5cm] (matching-x-y) {Simplying using \autoref{lemma:rr-only-matching} to get $|X| = |Y|$};
\node [reduction, below of=anno-contr-vc, node distance=2.5cm] (constrained-dir-max-cut) {Instance of \textsc{Constrained Directed MaxCut}};
\node [algo, below of=constrained-dir-max-cut, node distance=2.25cm] (dp-algo) {Solve in $\calO^{\star}(2^{\calO(k)})$ using  \autoref{lemma:algo-const-dir-maxcut}};

% Draw edges
\path [line] (inputnode) -- node [anchor=east] {$k = \rank(G)$}(keqrankG);
\draw [line] (inputnode) -- node {$k < \rank(G)$}(invisiblenode);
\path [line] (invisiblenode) -- node {$2d \le k$}(klerankG2dlek);
\path [line] (invisiblenode) -- node {$d \le k < 2d$}(lemma-oct);
\path [line] (lemma-oct) |- node {} (yesinstance);
\path [line] (lemma-oct) -- node {} (vc-x-rank-d);
\path [line] (vc-x-rank-d) -- node [anchor=east] {\autoref{lemma:contr-VC-to-annot-contr-vc}} (anno-contr-vc);
\path [line] (anno-contr-vc) -- node {\autoref{lemma:rr-anno-contr-edges-X}} (input-bipartite);
\path [line] (input-bipartite) -- node {\autoref{lemma:red-const-contr-vc-maxcut}} (constrained-max-cut);
\path [line] (constrained-max-cut) -- node {} (is-k-eq-d);
\path [line] (is-k-eq-d) -- node {$k > d$} (xp-factor);
\path [line] (xp-factor) -| node [near start] {$k = d$} (matching-x-y);
\path [line] (is-k-eq-d) -- node {$k = d$} (matching-x-y);
\path [line] (matching-x-y) -- node {\autoref{lemma:redu-maxcut-digraph-maxcut}} (constrained-dir-max-cut);
\path [line] (constrained-dir-max-cut) -- node {} (dp-algo);
\end{tikzpicture}}
\end{center}
   \caption{Diagram of the algorithm for \textsc{Contraction(\vc)} given by \autoref{thm:algorithm}.
Recall that we can assume that $d \le k \leq \rank(G)$, hence the case distinction considered in the beginning is exhaustive. Note also that, in the case where $d \le k < 2d$, it holds that $\calO^{\star}(2^{\calO(k)})= \calO^{\star}(2^{\calO(d)})$.\label{fig:diagram-algo}}
\end{figure}
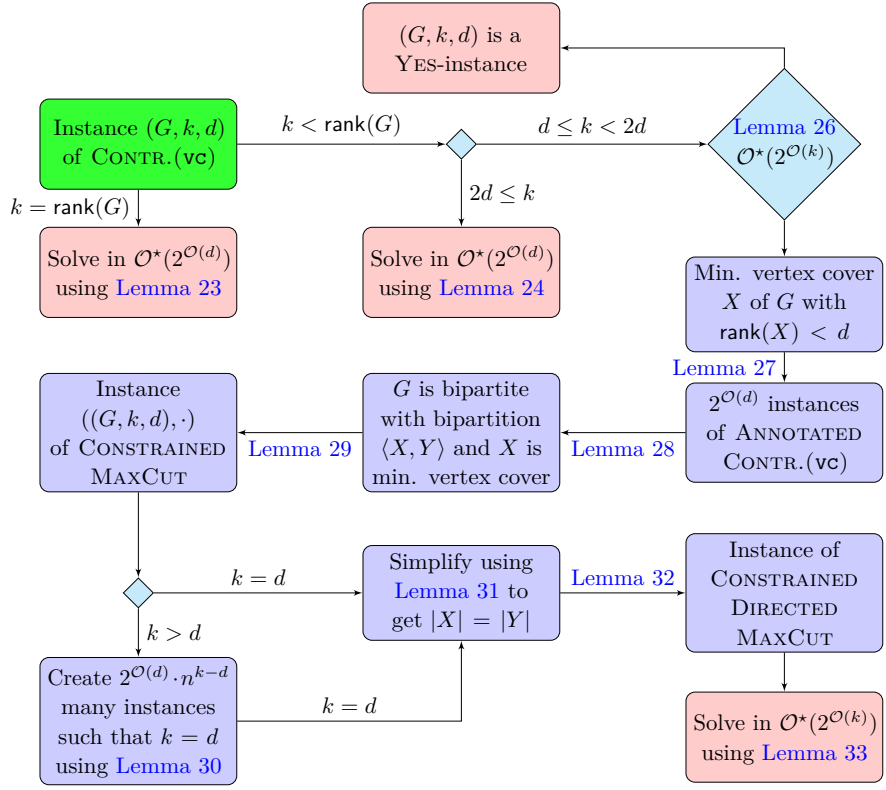

%\ptm{We need to mention \autoref{lemma:rr-anno-contr-edges-X} and \autoref{lemma:rr-only-matching} in the description or remove it from the figure.}

\iffalse
\begin{figure}[t]
  \begin{center}
    \includegraphics[scale=0.42]{./images/diagram-algo.JPG}
    \end{center}
   \caption{Diagram of the algorithm for \textsc{Contraction(\vc)} given by \autoref{thm:algorithm}. Recall that can assume that $d \le k \leq \rank(G)$, hence the three cases considered in the beginning are exhaustive. \label{fig:diagram-algo}}
\end{figure}
\fi

A solution of an instance of \textsc{Annotated Contraction(\vc)} is a solution pair $\langle X_s, Y_s \rangle$ as stated in \autoref{lemma:solution-edge-pair}.
We find convenient to present an algorithm that finds a partition $\langle V_L, V_R \rangle$ of $V(G)$ instead of a solution pair $\langle X_s, Y_s \rangle$.
To formalize this, we introduce the problem called \textsc{Constrained MaxCut} and we show it  to be equivalent to {\sc Annotated Contraction(\vc)} (cf. \autoref{lemma:red-const-contr-vc-maxcut}).
We partition the input instances of \textsc{Constrained MaxCut} into the following two types: $(1)$ $k = d$, and $(2)$ $k > d$. For the instances of the second type, we construct $2^{\calO(d)} \cdot n^{k - d}$ many instances of the first type such that the input instance is a \yes-instance if and only if at least one of these newly created instances is a \yes-instance (cf.
\autoref{lemma:k-neq-d-to-k-eq-d}).
We remark that this is the only step in the whole algorithm where a $n^{k-d}$-factor appears (note that this is unavoidable by \autoref{thm:w-hardness-edge-induced-forest}).

Finally, to handle the instances of the first type (i.e., with $k = d$), we first apply a simplification based on the existence of a matching saturating $X$ (cf. \autoref{lemma:rr-only-matching}), we introduce a directed variation of the problem called \textsc{Constrained Directed MaxCut}, and we prove it to be equivalent to its undirected version (cf. \autoref{lemma:redu-maxcut-digraph-maxcut}).
We then present a dynamic programming algorithm, with running time $2^{\calO(k)} \cdot n^{\calO(1)}$, that critically uses the fact that $k = d$ (cf. \autoref{lemma:algo-const-dir-maxcut}), in particular to ``merge'' appropriately some directed cycles in order to obtain a directed \emph{acyclic} graph, whose topological ordering gives a natural way to process the vertices of the input graph in a dynamic programming fashion.
At the end of \autoref{subsubsec-simplify-k-eq-d} we  present an overview of this algorithm.

\medskip

\noindent \textbf{Organization}.
In \autoref{sec:prelim} we present some notations, known results, preliminary results about \textsc{Contraction(\vc)}, and \autoref{lemma:solution-edge-pair} and \autoref{lemma:pendant-out-of-pair}.
In \autoref{sec:np-hardness} we present a reduction from a special case of \textsc{Multicolored Independent Set} to \textsc{Contraction(\vc)}.
This, along with other preliminary results, proves \autoref{thm:np-hard}.
In \autoref{sec:w-hard} we present two parameter preserving reductions, one from \textsc{Multicolored Independent Set} to \textsc{Edge Induced Forest}, and another one from \textsc{Edge Induced Forest} to \textsc{Contraction(\vc)}.
These two reductions, along with known results about \textsc{Multicolored Independent Set}, prove \autoref{thm:w1-hard}.
\autoref{sec:algorithm} is the most technical part of the paper, and contains the description of the algorithm to solve \textsc{Contraction(\vc)} mentioned in \autoref{thm:algorithm}.
We conclude the article in \autoref{sec:conclusion} with some open problems.
%Due to space limitations, in this extended abstract most of the technical contents of the paper  have been moved to the appendices, in particular the proofs of the statements marked with `$[\star]$'

\section{Preliminaries}
\label{sec:prelim}
For a positive integer $q$, we denote {the} set $\{1, 2, \dots, q\}$ by $[q]$.
We use $\mathbb{N}$ to denote the collection of all non-negative integers.

\subsection{Graph theory}
We use standard graph-theoretic notation, and we refer the reader to~\cite{Diestel12} for any undefined notation. For an undirected graph $G$, sets $V(G)$ and $E(G)$ denote its set of vertices and edges, respectively.
For a directed graph $D$,  sets $V(D)$ and $A(D)$ {denote} its set of vertices and arcs, respectively.
We denote an edge with two endpoints $u, v$ as $uv$.
To avoid confusion with {edges}, we denote an arc with {tail} $u$ and {head} $v$ as $(u, v)$.
Unless otherwise specified, we use $n$ to denote the number of vertices in the input (di)graph $G$ of the problem under consideration.
Two vertices $u, v$ in $V(G)$ are \emph{adjacent} if there is an edge $uv$ {in $G$}.
The \emph{open neighborhood} of a vertex $v$, denoted by $N_G(v)$, is the set of vertices adjacent to $v$.
The \emph{closed neighborhood} of a vertex $v$, denoted by $N_G[v]$, is the set $N_G(v) \cup \{v\}$.
We say that a vertex $u$ is a \emph{pendant vertex} if $|N_G(v)| = 1$.
The \emph{in-neighbourhood} of $v$ in {a} digraph $D$ is the set of vertices $u$ such that $(u, v)$ is an arc in $A(D)$.
We say {that} $(u, v)$ is an \emph{in-coming} arc of $v$.
Similarly, {the} \emph{out-neighbourhood} of $v$ is the set of vertices $u$ such that $(v, u)$ is an arc in $A(D)$.
We say {that} $(v, u)$ is an \emph{out-going} arc of $v$.
%For a vertex $v$, we define its in-deg and out-deg as size of its in-neighbourhood and out-neighborhood, respectively.
We denote the out-neighbors of $v$ by $N_{\sf out}(v)$.
We omit the subscript in the notation for neighborhood if the graph under consideration is clear.

For a subset $S$ of $V(G)$, we define $N[S] = \bigcup_{v \in S} N[v]$ and $N(S) = N[S] \setminus S$.
For a subset $F$ of edges, we denote by $V(F)$ the collection of endpoints of edges in $F$.
For a subset $S$ of $V(G)$ ({resp.} a subset $F$ of $E(G)$), we denote the graph obtained by deleting $S$ ({resp.} deleting $F$) from $G$ by $G - S$ ({resp.} by $G - F$).
We denote the subgraph of $G$ induced on the set $S$ by $G[S]$.
For two subsets $S_1, S_2$ of $V(G)$, $E(S_1, S_2)$ denotes the set of edges with one endpoint in $S_1$ and another one in $S_2$.
With {a slight abuse of notation}, we use $E(S_1)$ to denote the set $E(S_1, S_1)$.
Similarly, we define these notations for digraphs.
Namely,  $A(S_1, S_2)$ denotes the set of arcs with {tail} in $S_1$ and {head} in $S_2$.

A graph is {\em connected} if there is a path between every pair of distinct vertices.
A subset $S \subseteq V(G)$ is said to be a \emph{connected set} if $G[S]$ is connected.
A \emph{spanning tree} of a connected graph is {a} connected acyclic subgraph that includes all the vertices of the graph.
%By convention, the spanning tree of an isolated vertex is the empty set\igm{This is quite strange, as a spanning tree of an isolated vertex should be one vertex, not the empty set}.
A \emph{spanning forest} of a graph is a collection of spanning trees of its connected components.
The \emph{rank} of a graph $G$, denoted by $\rank(G)$, is the number of edges of a spanning forest of $G$ with the maximum number of edges.
The \emph{rank} of a set $X \subseteq V(G)$ of vertices, denoted by $\rank(X)$, is the rank of $G[X]$.
The \emph{rank} of a set $F \subseteq V(G)$ of edges, denoted by $\rank(F)$, is the rank of $V(F)$.
Note that an edge contraction decreases the rank of a graph $G$ by exactly one.

A set of vertices $Y$ is said to be {an} \emph{independent set} if no two vertices in $Y$ are adjacent.
We use the following observation.
\begin{observation}
\label{obs:rank-bound-vertices}
Consider two independent sets $X, Y$ in a graph $G$ such that there is no isolated vertex in $G[X \cup Y]$.
If $\emph{\rank}(E(X, Y)) \le k$, then $|X|,  |Y| \le k$.
\end{observation}
A graph is \emph{bipartite} if its vertex set can be partitioned into two independent sets.
For a graph $G$, a set $X \subseteq V(G)$ is said to be {a} \emph{vertex cover} if $V(G) \setminus X$ is an independent set.
A set of vertices $Y$ is said to be a \emph{clique} if any two vertices in $Y$ are adjacent.
A set of edges $M$ is called {a} \emph{matching} if no two edges in $M$ share an endpoint.
We say {that a} matching $M$ \emph{saturates} a set $X \subseteq V(G)$ if $X \subseteq V(M)$.
%We say an edge $e$ is \emph{covered by} vertex $x$ in vertex cover if $x$ is one of its endpoints.

A vertex cover $X$ is a \emph{minimum vertex cover} if for any other vertex cover $Y$ of $G$, we have $|X| \le |Y|$.
We denote by $\vc(G)$ the size of {a} minimum vertex cover of {a graph} $G$.
As {a} vertex cover needs to contain at least {one} vertex from each edge in a matching, $\vc(G)$ is at least the size of a maximum matching.
Consider a minimum vertex cover $X$.
For any $X' \subseteq X$, we have $|X'| \le  |N(X')|$ as otherwise $Y = (X \setminus X') \cup N(X')$ is another vertex cover of $G$ and $|Y| < |X|$,
contradicting the fact that $X$ is a minimum vertex cover.
As $|X'| \le  |N(X')|$ for every $X' \subseteq X$, Hall's theorem~\cite{Hall35} in bipartite graphs implies that there exists a matching saturating a minimum vertex cover $X$ in $G$.
Such a matching can be found in polynomial time \cite{HopcroftK73}. %\igm{CAREFUL!! All this is only for bipartite graphs}
For a graph $G$,  a set $X \subseteq V(G)$ is said to be {an} \emph{odd cycle transversal} if $G - X$ is a bipartite graph.
An odd cycle transversal $X$ is a \emph{minimum odd cycle transversal} if for any other odd cycle transversal $Y$ of $G$, we have $|X| \le |Y|$.
We denote by $\oct(G)$ the size of {a} minimum odd cycle transversal of {a graph} $G$.
We need the following algorithmic results regarding vertex covers and odd cycle transversals.
\begin{proposition}[\!\!\cite{chen2010improved}]
\label{prop:find-vc}
There is an algorithm that takes as input a graph $G$ and an integer $\ell$, runs in time $1.2738^{\ell}\cdot n^{\calO(1)}$, and correctly determines whether $\emph{\vc}(G) \le \ell$.
\end{proposition}
\begin{proposition}[Corollary~$10$ in \cite{narayanaswamy2012lp}]
\label{prop:find-oct}
There is an algorithm that takes as input a graph $G$ and an integer $\ell$, runs in time $2.6181^{\ell} \cdot n^{\calO(1)}$ and determines whether $\emph{\oct}(G) \le \ell$.
\end{proposition}
\begin{proposition}[Corollary~$15$ in \cite{narayanaswamy2012lp}]
\label{prop:find-vc-para-oct}
There is an algorithm that takes as input a graph $G$, runs in time $1.6181^{\emph{\oct(G)}}\cdot n^{\calO(1)}$ and computes a minimum vertex cover of $G$.
\end{proposition}

The algorithm in \autoref{prop:find-vc} can be easily modified to compute $\vc(G)$ if $\vc(G) \le \ell$.
For a graph $G$, we denote by $\bc(G)$ the minimum number of edges in $G$ that {need} to be contracted to make it a bipartite graph.
Note that for a set $F \subseteq E(G)$, if one can obtain a bipartite graph by contracting all edges in $F$, then one can obtain a bipartite graph by deleting one endpoint of every edge in $F$.
Hence,  we have the following observation.
\begin{observation}
\label{obs:oct-bound-bc}
For a graph $G$, $\emph{\oct}(G) \le \emph{\bc}(G)$.
\end{observation}
Consider a {(not necessarily proper)} $2$-coloring $\psi : V(G) \mapsto \{1, 2\}$.
Heggernes et al. \cite{heggernes2013obtaining} defined a notion of \emph{cost} of a $2$-coloring $\psi$ of a graph as $\sum_{X \in M_{\psi}}(|X| - 1)$, where $M_{\psi}$ is the set of monochromatic components of $\psi$.
Let $(V_1, V_2)$ be the partition of $V(G)$ such that every vertex in $V_1$ and  $V_2$ has color $1$ and $2$, respectively.
It is easy to see that cost of $\psi$ is equal to $\rank(V_1) + \rank(V_2)$.
We restate \cite[Lemma~$1$]{heggernes2013obtaining}
%\igm{\cite[Lemma 1]{heggernes2013obtaining}}
as the following observation.
\begin{observation}
\label{obs:bc-rank-equiv}
For a graph $G$, $\emph{\bc}(G) \le \ell$ if and only if there exists a  partition $(V_L, V_R)$ of $V(G)$ such that $\emph{\rank}(V_L) + \emph{\rank}(V_R) \le \ell$.
\end{observation}

\subsection{Edge contraction}
\label{prelim:edge-contraction}

The {\em contraction} of an edge $uv$ in a graph $G$ deletes vertices $u$ and $v$ from $G$, and adds a new vertex which is adjacent to all vertices that were adjacent to either $u$ or $v$.
This process does not introduce self-loops or parallel edges.
The resulting graph is denoted by $G/e$.
For a graph $G$ and edge $e = uv$, we formally define $G/e$ in the following way: $V(G/e) = (V(G) \cup \{w\}) \backslash \{u, v\}$ and $E(G/e) = \{xy \mid x,y \in V(G) \setminus \{u, v\}, xy \in E(G)\} \cup \{wx \mid x \in N_G(u) \cup N_G(v)\}$.
Here, $w$ is a new vertex.
An edge contraction reduces the number of vertices in a graph by exactly one.
Several edges might disappear because of one edge contraction.
For a subset of edges $F$ in $G$, graph $G/ F$ denotes the graph obtained from $G$ by contracting {all the edges in $F$}.

We now formally define a contraction of {a} graph $G$ to another graph $H$.
\begin{definition}[Graph contraction]\label{def:graph-contractioon} A graph $G$ is said to be \emph{contractible} to {a} graph $H$ if there is a function $\psi: V(G) \rightarrow V(H)$ such that {the} following properties hold.
\begin{enumerate}
\item For any vertex $h$ in $V(H)$, the set $W(h) := \{v \in V(G) \mid \psi(v)= h\}$ is not empty and the graph $G[W(h)]$ is connected.
\item For any two vertices $h, h’$ in $V(H)$, edge $hh’$ is present in $H$ if and only if $E(W(h), W(h’))$ is not empty.
\end{enumerate}
\end{definition}
We say {that} graph $G$ is \emph{contractible to $H$ via function $\psi$}.
For a vertex $h$ in $H$, the set $W(h)$, also denoted by $\psi^{-1}(h)$,  is called a \emph{witness set} associated with or corresponding to $h$.
For a fixed $\psi$, we define the $H$-\emph{witness structure} of $G$, denoted by $\mathcal{W}$, as {the} collection of all witness sets.
Formally, $\mathcal{W}=\{W(h) \mid h \in V(H)\}$.
{Note that} a witness structure $\mathcal{W}$ is a partition of {the} vertices in $G$.
If a {witness set} contains more than one vertex, then we call it {a} \emph{big} witness set, otherwise {we call it a} \emph{small} witness set.

\subsection{Contraction(\vc)}

In this subsection, we present a couple of observations regarding an instance $(G, k, d)$ of the \textsc{Contraction(\vc)} problem.
Later, we present a lemma that helps us to characterize the problem as finding a vertex cover with special properties.

\begin{observation}
\label{obs:bound-k-rank}
Consider an instance $(G, k, d)$ of {\sc Contraction(\vc)} such that $k = \emph{\rank}(G)$.
Then,  $(G, k, d)$ is a {\sc Yes}-instance if and only if $d \le \emph{\vc}(G)$.
\end{observation}
\begin{proof}
Suppose that $(G, k,d)$ is a \yes-instance.
Let $F \subseteq E(G)$ be {a} collection of at most $k$ edges in $G$ such that $\vc(G/F) \le \vc(G) - d$.
As $\vc(G/F) \ge 0$,  we have $d \le \vc(G)$.

Suppose now that $d \le \vc(G)$.
Let $F$ be the collection of edges in a spanning forest of $G$.
Note that the graph $G/F$ does not contain any edge and hence $\vc(G/F) = 0 \le \vc(G) - d$.
As $k = \rank(G)$,  we have $|F| = k$.
Hence,  $F$ is a solution of $(G, k, d)$.
\end{proof}

\begin{observation}
\label{obs:bound-d-k}
Consider an instance $(G, k, d)$ of {\sc Contraction(\vc)} such that $G$ is a connected graph,  $k < \emph{\rank}(G)$, and $2d \le k$.
Then,  $(G, k, d)$ is a {\sc Yes}-instance if and only if $d < \emph{\vc}(G)$.
\end{observation}
\begin{proof}
Suppose that $(G, k, d)$ is a \yes-instance.
Let $F \subseteq E(G)$ be a collection of at most $k$ edges in $G$ such that $\vc(G/F) \le \vc(G) - d$.
As $|F| \le k$ and $k < \rank(G)$, the graph $G/F$ contains at least one edge.
Hence, $\vc(G/F) \ge 1$.
This implies $1 + d \le \vc(G)$.

Suppose now that $d < \vc(G)$.
Let $X$ be a minimum vertex cover of $G$.
Consider an algorithm that contracts a path between two vertices in $X$ that are distance at most two.
The existence of such vertices is guaranteed by the fact that $G$ is a connected graph.
Let $G'$ be the resulting graph.
Note that $\vc(G') = \vc(G) - 1$.
It is easy to verify that if $(G, k, d)$ is a \yes-instance, then $(G', k - 2,  d - 1)$ is a \yes-instance.
As $d < \vc(G) $ and $k \ge 2d$, the subroutine can repeat the process $d$ times to get an equivalent instance $(G', k', d')$ such that $k' \ge 0$ and $d' = 0$.
As $(G', k', d')$ is a trivial \yes-instance, $(G, k, d)$ is a \yes-instance.
\end{proof}

\begin{figure}[t]
  \begin{center}
    \includegraphics[scale=0.53]{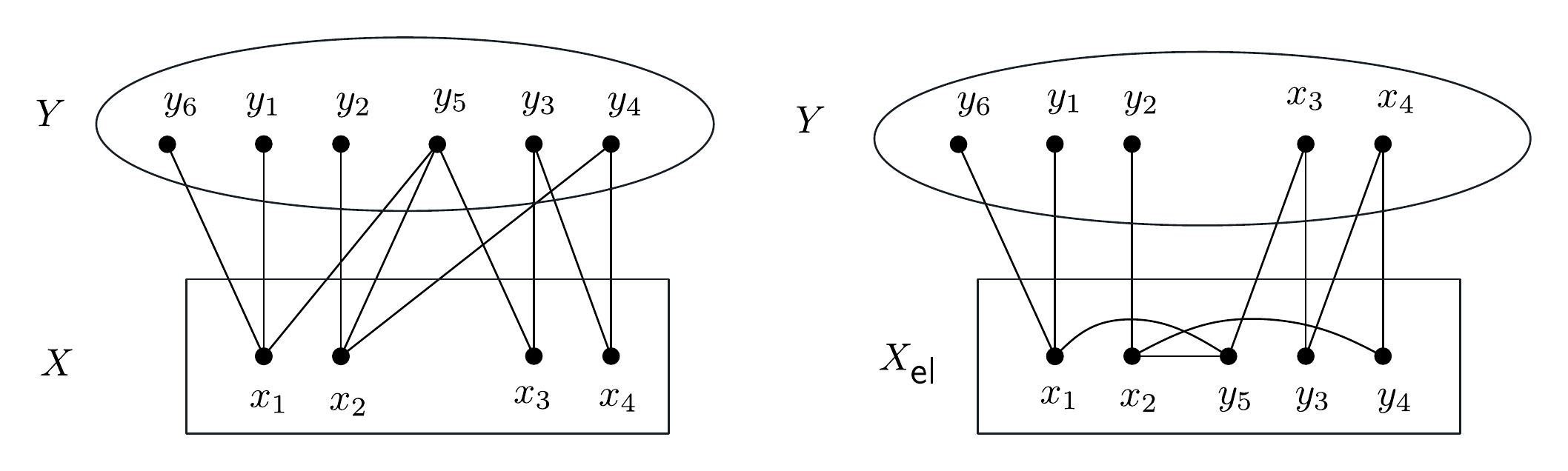}
    \end{center}
   \caption{We can reduce the size of {a} minimum vertex cover of $G$ by two by contracting {the} three edges in $F = \{x_1y_5, x_2y_5, x_2y_4\}$, i.e., $\vc(G/F) \le \vc(G) - d$ for $d = 2$.
\autoref{lemma:solution-edge-pair} implies that there exists a solution pair $\langle X_s = \{x_3, x_4\}, Y_s = \{y_3, y_4, y_5\} \rangle$ such that $\rank((X \setminus X_s) \cup Y_s) = \rank(X_{\el}) \ge 3 = |F|$ and $|Y_s| - |X_s| \le 1 = |F| - d$. \label{fig:sol-edge-sol-pair}}
\end{figure}

Suppose that $(G, k, d)$ is a \yes-instance of \textsc{Contraction(\vc)}.
We say {that} a set $F \subseteq E(G)$ is a \emph{solution} of $(G, k, d)$ if $|F| \le k$ and $\vc(G/F) \le \vc(G) - d$.
Fix a minimum vertex cover $X$ of $G$.
As $X$ is a vertex cover, for every edge in $F$, \emph{at least one} of its endpoints is in $X$.
We argue that one can construct an \emph{enlarged vertex cover} $X_{\el}$ of $G$ such that for every edge in $F$, \emph{both} of its endpoints are in $X_{\el}$.
Also, $X_{\el}$ is not much larger than $X$.
In order to construct $X_{\el}$ from $X$,  one needs to remove and add  some vertices to $X$.
We denote {the} removed and added vertices by $X_s$ and $Y_s$, respectively, and call $\langle X_s, Y_s \rangle$ a \emph{solution pair}.
See \autoref{fig:sol-edge-sol-pair} for an illustration.
The following lemma relates a solution (a set of edges) to a solution pair (a tuple of disjoint vertex sets).

\begin{lemma}
\label{lemma:solution-edge-pair}
Consider a connected graph $G$,  a minimum vertex cover $X$ of $G$,  a proper subset $F$ of edges of a spanning forest of $G$ (i.e., $|F| < \emph{\rank}(G)$), and {a non-negative} integer~$d$.
Then,  $\emph{\vc}(G/F) \le \emph{\vc}(G) - d$ if and only if there exists subsets $X_s \subseteq X$ and $Y_s \subseteq V(G) \setminus X$ such that
$(i)$ $X_{\emph{\el}} := (X \setminus X_s) \cup Y_s$ is a vertex cover of $G$,
$(ii)$ $\emph{\rank}((X \setminus X_s) \cup Y_s) \ge |F|$, and
$(iii)$ $ |Y_s| - |X_s|  \le |F| - d$, i.e., $|X_{\emph{\el}}| \le |X| + |F| - d$.
\end{lemma}
\begin{proof}
$(\Rightarrow)$ Consider the collection $\calF$ of subsets  of  $E(G)$ such that for every $F \in \calF$,  $\vc(G/F) \le \vc(G) - d$.
For $F \in \calF$, suppose $G$ is contracted to $G/F$ via function $\psi$.
Let $X'$ be a minimum vertex cover of $G/F$ and $Y' = V(G/F) \setminus X'$.
We say {that} $\langle X', Y' \rangle$ is a partition corresponding to $F$.
We define a function $\cost: \calF \mapsto \mathbb{N}$ as follows.
For $F \in \calF$, $\cost(F)$ is the minimum number of vertices in $Y'$ that are associated with big witness sets over all partitions $\langle X', Y' \rangle$ corresponding to $F$.
Formally, $\cost(F) := \min_{\langle X', Y' \rangle}|\{y \in Y'\ |\ |\psi^{-1}(y)| > 1\}|${, where $\langle X', Y' \rangle$ ranges over all partitions corresponding to $F$.}
%We argue that there is a set in $\calF$ whose cost is zero.

We assume, for the sake of contradiction, that there is no set in $\calF$ whose cost is zero.
Let $F \in \calF$ be a set of edges of minimum cost over all sets in $\calF$.
By our assumption,  $\cost(F) > 0$.
This implies that there is a partition $\langle X', Y' \rangle$ of $V(G/F)$ and  a vertex $y' \in Y'$ such that $|\psi^{-1}(y')| > 1$.
Recall that $F$ is a proper subset of edges in a spanning forest of $G$.
Hence,  $|F| < \rank(G)$ and there is at least one edge present in $G/F$.
This implies that a minimum vertex cover $X'$ of $G/F$ is not empty.
As $G$ is a connected graph, so is $G/F$.
Hence,  there is a vertex $x' \in X'$ such that $x'y' \in E(G/F)$.

Consider {a} $G/F$-witness structure $\calW$ of $G$.
Let $W(x'), W(y')$ be the witness sets corresponding to $x'$ and $y'$, respectively.
Recall that $W(x')$ and $W(y')$ are connected sets in $G$.
As $x'y' \in E(G/F)$, there exists an edge $e$ in $E(G)$ with one of its endpoints in $W(x')$ and {the} other in $W(y')$.
Hence, $W(x') \cup W(y')$ is a connected set in $G$.
We claim that there is a spanning tree of $G[W(x') \cup W(y')]$ that has a leaf in $W(y')$.

Let $T_x$ and $T_y$ be spanning trees of $G[W(x')]$ and $G[W(y')]$, respectively.
Without loss of generality, we can assume that $E(T_x) \cup E(T_y) \subseteq F$.
As $|W(y')|>1$,  $T_y$ contains at least one edge and hence at least two leaves.
Consider the tree $T_{xy}$ such that $E(T_{xy}) = E(T_x) \cup \{e\} \cup E(T_y)$, where $e$ is the edge mentioned in the previous paragraph.
Note that $T_{xy}$ is a spanning tree of $G[W(x') \cup W(y')]$.
This spanning tree has a leaf, say $y_1$,  in $W(y')$,  as $e$ is incident on at most one leaf of $T_y$.

Consider the partition $\calW_1$ obtained from $\calW$ by removing $W(x'), W(y')$ and adding $W_{x^{\circ}}, W_{y^{\circ}}$.
Here $W_{x^{\circ}} = (W(x') \cup W(y')) \setminus \{y_1\}$ and $W_{y^{\circ}} = \{y_1\}$.
Formally,  $\calW_1 = (\calW \cup \{W_{x^{\circ}} , W_{y^{\circ}} \} )\setminus \{W(x'), W(y')\}$.
Let $F_1 = (F \cup E(T_{xy})) \setminus (E(T_x) \cup E(T_y)) $.
It is easy to verify that $\calW_1$ is a $G/F_1$-witness structure of $G$.
As $F_1$ is obtained from $F$ by removing an edge incident on $y_1$ and adding edge $e$ (which was not in $F$), we have $|F_1| = |F|$.
Also, note that $\cost(F_1) < \cost(F)$ as the witness set corresponding to $y'$ no longer contributes to the cost.

We argue that $F_1$ is  in $\calF$.
Let $x^{\circ}$ and $y^{\circ}$ be the two vertices corresponding to witness sets $W_{x^{\circ}}$ and $W_{y^{\circ}}$, respectively.
Note that the graph obtained from $G/F$ by deleting vertices {$x$ and $y$} is {the} same as the graph obtained from $G/F_1$ by deleting vertices $x^{\circ}$ and $y^{\circ}$.
As $W(x) \subseteq W_{x^{\circ}}$,  $x^{\circ}$ covers all the edges in $G/F$ that were covered by $x$.
Also, as $y$ was not in a vertex cover, it did not cover any edge in $G/F$.
This implies $\vc(G/F) = \vc(G/F_1)$.
Thus, $\vc(G/F_1) \le \vc(G) - d$,  and $F_1$ is in $\calF$.
But this contradicts the fact that $F$ is {a} set of edges with minimum cost.
Hence, our assumption was wrong and there exists a set of edges in $F$ whose cost is zero.

Consider a set $F \in \calF$ such that $\cost(F) = 0$.
Let $\langle X', Y' \rangle$ be the partition of $V(G/F')$ such that $X'$ is a vertex cover of $G$ and for every $y' \in Y'$,  $|\psi^{-1}(\{y'\})| = 1$.
Let $X_{\el} = \bigcup_{x' \in X'} \psi^{-1}(\{x'\})$.
Alternately, $X_{\el}$ is the subset of vertices in $V(G)$ such that $\psi(X_{\el}) = X'$.
Define $X_s := X \setminus X_{\el}$ and $Y_s := X_{\el} \cap Y$.
We argue that $\langle X_s, Y_s \rangle$ is a solution pair.
As $|\psi^{-1}(\{y'\})| = 1$ {for every $y' \in Y'$}, $X_{\el} = (X \setminus X_s) \cup Y_s$ is a vertex cover of $G$.
Recall that $F$ is a subset of edges in a spanning forest of $G$.
As $V(F) \subseteq X_{\el}$, the rank of $X_{\el} = (X \setminus X_s) \cup Y_s$ is at least $|F|$.
Note that the set $X'$ can be obtained from $X_{\el}$ by contracting the edges in $F$.
As $F$ is a subset of edges of a forest, we get $|X_{\el}| \le |X'| + |F| \le |X| - d + |F|$.
Hence, $\langle X_s, Y_s \rangle$ satisfies all three properties.

$(\Leftarrow)$
Suppose that there is a solution pair $\langle X_s, Y_s \rangle$ such that $X_s \subseteq X$ and $Y_s \subseteq Y$ with the desired properties.
Define $X_{\el}  = (X \setminus X_s) \cup Y_s$.
Let $F$ be the {edge set of} a spanning forest of $G[X_{\el}]$ such that $|Y_s| - |X_s| \le |F| - d$.
Such a set of edges exists as $\rank(X_{\el}) \ge |F|$.
Hence, $|X_{\el}| = |X| - |X_s| + |Y_s| \le |X| + |F| - d$.
As $X_{\el}$ is a vertex cover of $G$, set $V(G[X]/F)$ is a vertex cover of $G/F$.
As $F$ is  the {edge set of} a spanning forest of $G[X_{\el}]$, we have $\vc(G/F) \le |V(G[X]/F)| = |X_{\el}| - |F| = |X| + (|F| - d) - |F| \le |X| - d = \vc(G) - d$.
Hence, $\vc(G/F) \le \vc(G) - d$.
\end{proof}

In the following lemma, we argue that there exists a solution pair $\langle X_s, Y_s \rangle$ such that $X_s$ does not  {contain} any vertex in $X$ which is adjacent to a pendant vertex.
For example, in \autoref{fig:sol-edge-sol-pair}, there exists a solution pair $\langle X_s, Y_s \rangle$ such that $x_1 \not\in X_s$.

\begin{lemma}
\label{lemma:pendant-out-of-pair}
Consider a connected graph $G$, a minimum vertex cover $X$ of $G$, and two integers $\ell$ and $d$.
Suppose that there exists a vertex $x^{\circ}$ in $X$ which is adjacent to a pendant vertex.
Suppose that there are subsets $X_s \subseteq X$ and $Y_s \subseteq V(G) \setminus X$ such that
$(i)$ $(X \setminus X_s) \cup Y_s$ is a vertex cover of $G$,
$(ii)$ $\emph{\rank}((X \setminus X_s) \cup Y_s) \ge \ell$, and
$(iii)$ $ |Y_s| - |X_s|  \le \ell - d$.
Then, there are subsets $X'_s \subseteq X$ and $Y'_s \subseteq V(G) \setminus X$ that satisfy these three conditions and $x^{\circ} \not\in X'_s$.
\end{lemma}
\begin{proof}
If $x^{\circ} \not\in X_s$ then the lemma is vacuously true.
Consider the case {where $x^{\circ} \in X_s$}.
Let $y^{\circ}$ be a pendant vertex in $G$ which is adjacent to $x^{\circ}$.
As $(X \setminus X_{s}) \cup Y_s$ is a vertex cover of $G$, $y^{\circ}$ is in it.
More specifically, $y^{\circ} \in Y_s$.
Define $X'_s := X_s \setminus \{x^{\circ}\} $ and $Y'_s := Y_s \setminus \{y^{\circ}\}$.

As $(X \setminus X_{s}) \cup Y_s$ is a vertex cover of $G$, and $y^{\circ}$ is a pendant vertex, {it follows} that $(X \setminus (X_{s} \cup \{x^{\circ}\}) \cup (Y_s \setminus \{y^{\circ}\})$ is also a vertex cover of $G$.
As $y^{\circ}$ is not adjacent to any vertex in $(X \setminus X_{s}) \cup Y_s$, we have $\rank((X \setminus X_s) \cup (Y_s \setminus \{y^{\circ}\})) = \rank((X \setminus X_s) \cup Y_s) \ge \ell$.
Removing a vertex from $X_S$, which is the same as adding a vertex in $(X \setminus X_s) \cup (Y_s \setminus \{y^{\circ}\}))$, cannot decrease its rank.
This implies $\rank((X \setminus X'_s) \cup Y'_s)) \ge \ell$.
Note that $|X'_s| = |X_s| -  1$ and $|Y'_s| = |Y_s| - 1$.
Hence, $|Y'_s| - |X'_s| \le \ell - d$.
As $\langle X'_s, Y'_s \rangle$ satisfies all the three properties, and $x^{\circ} \not\in X'_s$, we get a solution pair with the desired properties.
\end{proof}

\subsection{Parameterized complexity}
\label{prelim:pc}

An instance of a parameterized problem $\Pi$ {consists} of an input $I$, which is an input of the non-parameterized version of the problem, and an integer $k$, which is called the \emph{parameter}.
A problem $\Pi$ is said to be \emph{fixed-parameter tractable}, or \FPT, if given an instance $(I,k)$ of $\Pi$, we can decide whether  $(I,k)$ is a \yes-instance of $\Pi$ in  time $f(k)\cdot |I|^{\OO(1)}$.
Here, {$f: \mathbb{N} \mapsto \mathbb{N}$} is some computable function {depending} only on $k$.
Parameterized complexity theory  provides tools to {rule out} the existence of \FPT\ algorithms under plausible complexity-theoretic assumptions.
For this, a hierarchy of parameterized complexity classes $\FPT \subseteq \W[1]\subseteq \W[2] \cdots \subseteq \XP$ was introduced, and it was conjectured that the inclusions are proper.
The most common way to show that it is unlikely that a parameterized problem {admits} an \FPT\ algorithm is to show that it is $\W[1]$ or $\W[2]$-\Hard.
It is possible to use reductions analogous to the polynomial-time reductions employed in classical complexity.
Here, the concept of \W$[1]$-\Hardness\ replaces the one of \NP-\Hardness, and we need not only {to} construct an equivalent instance  \FPT\ time, but also {to} ensure that the size of the parameter in the new instance depends only on the size of the parameter in the original instance.
These types of reductions are called \emph{parameter preserving reductions}.
For {a} detailed introduction to parameterized complexity and related terminologies, we refer the reader to the recent books by Cygan et al.~\cite{cygan2015parameterized} and Fomin et al.~\cite{fomin2019kernelization}.

In {the} \textsc{Multicolored Independent Set} {problem}, the input is a graph $G$, an integer $q$, and  a partition $\langle V_1, V_2, \dots, V_q \rangle$ of $V(G)$.
The  objective is to determine whether there exists a multicolored independent set in $G$. We say {that} an independent set  in $G$ is \emph{multicolored} if it contains one vertex from $V_i$ for every $i \in [q]$.
Note that it is safe to assume that each $V_i$ is a clique in $G$.
We will need the following result.
\begin{proposition}[cf. Theorem~$14.21$ in  \cite{cygan2015parameterized}]
\label{prop:ind-set-w-hard}
{\sc Multicolored Independent Set} parameterized by the size of the solution $q$ is \emph{\W[1]-\Hard}.
Moreover, unless {the} \emph{\ETH} fails, it does not admit an algorithm running in time $f(q) \cdot n^{o(q)}$ for any computable function {$f: \mathbb{N} \mapsto \mathbb{N}$}.
\end{proposition}

A \emph{reduction rule} is a polynomial-time  algorithm that takes as input an instance of a problem and outputs another, usually reduced, instance.
A reduction rule said to be \emph{applicable} on an instance if the output {instance} and input instance are different.
A reduction rule is \emph{safe} if the input instance is a \yes-instance if and only if the output instance is a \yes-instance.

\section{\NP-hardness results}
\label{sec:np-hardness}

In this section we prove \autoref{thm:np-hard}.
The first and the second item in the statement of \autoref{thm:np-hard} {follow} directly from \autoref{obs:bound-k-rank} and \autoref{obs:bound-d-k}, respectively.
Hence,  we focus on the third case in this section.
Recall that in the \textsc{Multicolored Independent Set} problem, the input is a graph $G$,  an integer $q$, and a partition $\langle V_1, V_2, \dots, V_q \rangle$ of $V(G)$.
We consider a special case of this problem and call it $(3 \times q)$-\textsc{Multicolored Independent Set}.
In this problem, the input is {the} same as that of \textsc{Multicolored Independent Set}, but it comes with a guarantee that every color class has exactly three vertices.
The \NP-\Hardness\ of this problem follows from the standard reduction from $3$-\textsc{SAT} to \textsc{Independent Set} ({see, for example,~\cite[Theorem~$8.8$]{kleinberg2006algorithm}}).
This reduction ensures that each color class is a clique of size two or three.
For every color class $V_i$ that contains two vertices, we add  a new vertex to $V_i$ and make it adjacent to every vertex in the graph.

\begin{figure}[t]
  \begin{center}
    \includegraphics[scale=0.58]{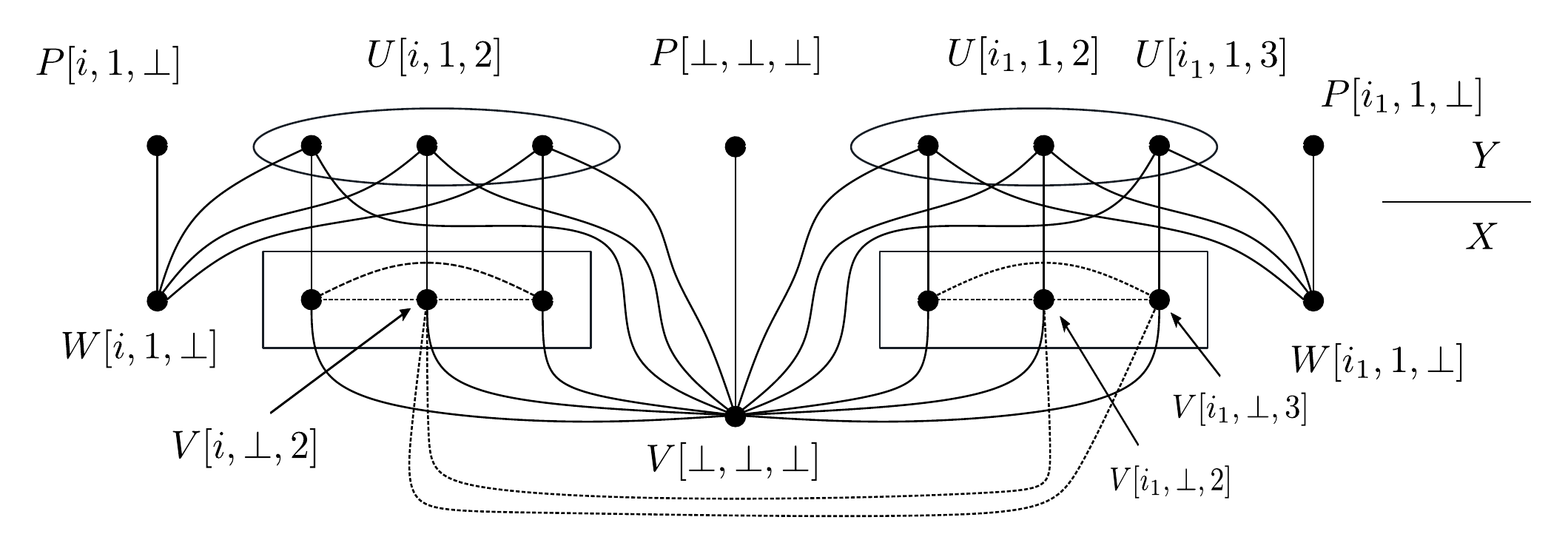}
    \end{center}
   \caption{Reduction from $(3 \times q)$-\textsc{Multicolored Clique} to \textsc{Contraction(\vc)} for  $\ell = 1$. Every vertex below the horizontal line (accompanied by $X, Y$) on right is in $X$ which is a minimum vertex cover of $G'$. Dashed edges shows edges in $G$. \label{fig:np-hard}}
\end{figure}

\medskip
 \noindent \textbf{The reduction:}
The reduction takes as input an instance $(G, q, \langle V_1, V_2, \dots, V_q \rangle)$ of $(3 \times q)$-\textsc{Multicolored Independent Set}, a positive integer $\ell$,  and returns an instance $(G', k, d)$ of \textsc{Contraction(\vc)} such that $k = d + \frac{\ell - 1}{\ell + 3} \cdot d$.
For notational convenience, rename the {three} vertices {in the} $i^{th}$ color class {of} $G$ {as}  $V[i, \perp, 1], V[i,  \perp, 2]$, and $V[i, \perp, 3]$ for every $i \in [q]$.
We use a similar notation to refer to new vertices added to $G$ in order to construct $G'$.
We use $i$ and $j$ as the running variables in the set $[q]$ and $[\ell]$, respectively.
See \autoref{fig:np-hard} for an illustration for the case when $\ell = 1$.
The reduction adds the following vertices to a copy of $G$ to construct $G'$:

\begin{itemize}
\item $W[i, j,  \perp]$ and $P[i, j, \perp]$ for every $i \in [q]$  and every $j \in [\ell]$,
\item $U[i, j, 1], U[i, j, 2], $ and $U[i, j, 3]$,   for every $i \in [q]$ and every $j \in [\ell]$, and
\item two vertices denoted by $V[\perp, \perp, \perp]$ and $P[\perp, \perp, \perp]$.
\end{itemize}

It adds the following edges{:}

\begin{itemize}
\item For every $i \in [q]$ and $j \in [\ell]$,  it adds the four edges incident on $W[i, j, \perp]$ whose other endpoints are $P[i, j, \perp]$,  $U[i, j, 1]$,  $U[i, j, 2]$, and $U[i, j, 3]$.
\item For every $i \in [q]$ and $j \in [\ell]$, it adds three matching edges whose endpoints are $\{V[i,\perp, 1], U[i, j, 1]\}$, $\{V[i, \perp, 2], U[i, j, 2]\}$, and $\{V[i, \perp, 3], U[i, j, 3]\}$.
\item For {every} $i \in [q]$ and $j \in [\ell]$, it adds three edges incident on $V[\perp, \perp, \perp]$ whose other endpoints are $V[i, \perp, 1]$, $V[i, \perp, 2]$,  and $V[i, \perp, 3]$.
It {adds} three more edges incident on $V[\perp, \perp, \perp]$ whose other endpoints are $U[i, j, 1]$, $U[i, j, 2]$,  and $U[i, j, 3]$.
\item It adds {an} edge with endpoints $V[\perp, \perp, \perp ]$ and $P[\perp,\perp, \perp]$.
\end{itemize}

This completes the construction of $G'$.
The reduction sets $d = (\ell + 3) \cdot q$ and $k = d + (\ell - 1) \cdot q$, and {returns}  $(G', k , d)$ as the instance of \textsc{Contraction(\vc)}.

We define sets $V, U, W$, and $P$ in the natural way, i.e., $V$ is the collection of all the vertices that have representation $V[i, j, \perp]$ for some $i\in [q]$ and $j \in [\ell]$.
We define {the} other sets {similarly}.
Note that $V[\perp, \perp, \perp] \not\in V$ and $P[\perp, \perp, \perp] \not\in P$.
By the construction,  every vertex in $\{P[\perp, \perp, \perp]\} \cup P$ is a pendant vertex.

For the sake of simplicity, we start by presenting an overview of the correctness of the reduction for the case where $\ell = 1$, i.e., $k = d$.
The formal proof is provided after the overview.
%In this reduced instance, we want to find a subset $F$ of size at most $k$ such that $\vc(G/F) \le \vc(G) - d$.
By \autoref{lemma:solution-edge-pair}, there is a solution $F$ of $(G', k, d)$ if and only if there exists a solution pair $\langle X_{s}, Y_s \rangle$ such that
$(i)$ $X_{\el} = (X \setminus X_{s}) \cup Y_s$ is a vertex cover of $G'$,
$(ii)$ $\rank(X_{\el}) \ge |F| = k$, and
$(iii)$ $|X_{\el}| \le |X| + |F| - d \le |X| + k - d = |X|$.
The reduction ensures that the size of $X$ is $k + 1$.
With this, the second and the third conditions force $X_{\el}$ to be a connected set of the same size as that of $X$.
For the example in \autoref{fig:np-hard}, \autoref{lemma:pendant-out-of-pair} implies that $V[\perp, \perp, \perp]$, $W[i, 1, \perp]$ and $W[i_1, 1, \perp]$ are in $X_{\el}$.
Hence, to provide connectivity between $V[\perp, \perp, \perp]$ and $W[i, 1, \perp]$, at least one of the vertices in $\{U[i, 1, 1], U[i, 1, 2], U[i, 1, 3]\}$ needs to be in $X_{\el}$.
However, as $|X_{\el}| = |X|$, at least one vertex in $\{V[i, \perp, 1], V[i, \perp, 2], V[i, \perp, 3]\}$ needs to be out of $X_{\el}$, i.e., in $X_s$.
As this is true for every color class, $X_{s}$ includes at least one vertex from it.
The first condition enforces $X_{s}$ to be an independent set.
This implies that $X_s$ can include at most one vertex from each color class.
Moreover, if $X_s$ includes $V[i, \perp, 2]$ then it cannot include $V[i_1, \perp, 2]$ or $V[i_1, \perp, 3]$.
These are precisely the conditions we want for encoding an instance of \textsc{Multicolored Independent Set}.
This concludes the overview of the reduction.

We formalize the above ideas in \autoref{lemma:np-hard-forward} and \autoref{lemma:np-hard-backward}.
Before that, in the next lemma we argue about the size of a minimum vertex cover of $G'$.

\begin{lemma}
\label{lemma:np-hard-vc}
The set $X := V \cup W \cup \{V[\perp, \perp, \perp] \}$ is a minimum vertex cover of $G'$.
\end{lemma}
\begin{proof}
By the construction of $G'$,  it is easy to verify that $X$ is a vertex cover of $G'$.
To prove that it is a minimum vertex cover, we show that there is a matching $M$ of size $|X|$ in $G'$.
Initialize $M = \emptyset$.
For every vertex in $\{V[\perp, \perp, \perp]\} \cup W$, include the edge in $M$ incident on its pendant neighbor.
For every $i \in [q]$, include the three edges whose endpoints are $\{V[i, \perp, 1], U[i, j, 1]\}$, $\{V[i, \perp, 2], U[i, j, 2]\}$, and $\{V[i, \perp, 3], U[i, j, 3]\}$.
It is easy to verify that $M$ is a matching of size $|X|$.
This implies that $X$ is a minimum vertex cover of $G'$.
\end{proof}

%Note that the size of this vertex cover is $(\ell + 3) \cdot q + 1 = d + 1$.
%Hence, $\vc(G') = d + 1$.
%We argue that in this case, it is sufficient and necessary to find a vertex cover $X_{\el}$ of $G'$ that contains $\{V[\perp, \perp, \perp]\} \cup W$,  is of size at most $|X| + (k - d) = |X| + (\ell - 1) \cdot q$, and is connected.
%As mentioned before, in order to provide connectivity between $V[\perp, \perp, \perp]$ and $W[i, j, \perp]$,  at least one of the vertices among $U[i, j, 1]$, $U[i, j, 2]$, or $U[i, j, 3]$ must be present in $X_{\el}$.
%As this statement is true for every $i \in [q]$ and every $j \in [q]$,  $X_{\el}$ contains at least $\ell \cdot q$ many vertices in $U$.
%However, $|X_{\el}|$ is at most $|X| + (k - d) = |X| + \ell \cdot q - q$.
%Hence, at least $q$ vertices in $V$ {need} to be shifted out of $X$ in order to create $X_{\el}$.
%We argue that these vertices {correspond} to a multicolored clique in $G$.

\begin{lemma}
\label{lemma:np-hard-forward}
If $(G, q, \langle V_1, V_2, \dots, V_q \rangle)$ is a {\sc Yes}-instance of $(3 \times q)$-{\sc Multicolored Independent Set},  then $(G', k, d)$ is a {\sc Yes}-instance of {\sc Contraction(\vc)}.
\end{lemma}
\begin{proof}
Suppose that $Q$ is a multicolored independent set in $G$.
Let $\{V[i, \perp, z_i]\} = Q \cap V_i$ for $z_i \in \{1, 2, 3\}$.
By \autoref{lemma:np-hard-vc},  $X := V \cup W \cup \{V[\perp, \perp, \perp] \}$ is a minimum vertex cover of $G'$.
Define
$$X_s := \{V[i, \perp, z_i]\ |\ i \in [q] \},\ Y_s := \{U[i, j, z_i]\ |\  i \in [q]\ \land\ j \in [\ell]\}\ \text{and}\ X_s := X \setminus X_{\el}.$$
It is easy to verify that $X_{\el}$ is a vertex cover of $G'$.
As $G[X_{\el}]$ is a connected graph, $\rank(X_{\el}) = |X_{\el}| - 1 = |X| + (\ell - 1) \cdot q - 1$.
As  $|X| = (\ell + 3) \cdot q + 1  = d + 1$,   we get $\rank(X_{\el}) = d + 1 + (\ell - 1) \cdot q - 1 = k$.
Also, $|Y_s| -  |X_s| = (\ell - 1) \cdot q = k - d$.
This implies that the pair $\langle X_s, Y_s \rangle$ satisfies all the three conditions mentioned in the statement of \autoref{lemma:solution-edge-pair}.
As $k < \rank(G)$,  \autoref{lemma:solution-edge-pair} implies that there exists a set of {edges} $F$ of size at most $k$ in $G'$ such that $\vc(G'/F) \le \vc(G) - d$.
Hence, $(G', k, d)$ is a {\yes-instance} of \textsc{Contraction(\vc)}.
\end{proof}

\begin{lemma}
\label{lemma:np-hard-backward}
If $(G', k, d)$ is a {\sc Yes}-instance of {\sc Contraction(\vc)} then $(G, q, \langle V_1, V_2, \dots, V_q \rangle)$ is a {\sc Yes}-instance of $(3 \times q)$-{\sc Multicolored Independent Set}.
\end{lemma}
\begin{proof}
Suppose that $F'$ is a solution of $(G', k, d)$, i.e., $\vc(G'/F) \le \vc(G) - d$ and $|F'| \le k$.
As $k < \rank(G)$, we can assume, without loss of generality, that $|F'| = k$.
\autoref{lemma:solution-edge-pair} implies that there exists a solution pair $\langle X_s, Y_s \rangle$ that satisfies the three conditions mentioned in its statement.
Recall that every vertex in $\{V[\perp, \perp, \perp]\} \cup W$ is adjacent to some pendant vertex in $G'$.
\autoref{lemma:pendant-out-of-pair} implies that there exits a solution pair $\langle X_s, Y_s \rangle$ with {the} additional property that $X_s \cap (\{V[\perp, \perp, \perp]\} \cup W) = \emptyset$.
This implies $(\{V[\perp, \perp, \perp]\} \cup W) \subseteq X_{\el}  := (X \setminus X_s) \cup Y_s$.

We first argue that $q \le |X_s|$.
By the second condition in \autoref{lemma:solution-edge-pair},  $\rank(X_{\el}) \ge k$.
By the third condition in  \autoref{lemma:solution-edge-pair}, $|X_{\el}| \le  |X| + k - d$.
As $|X| = d + 1$, {it follows} that $|X_{\el}| \le k + 1$.
Hence, the number of vertices in $G'[X_{\el}]$ is at most $k + 1$, whereas the number of edges in a spanning forest of $G'[X_{\el}]$ is at least $k$.
This implies that $G'[X_{\el}]$ is connected.
Fix integers $i \in [q]$ and $j \in [\ell]$.
By the construction of $G'$,  every path between $V[\perp, \perp, \perp]$ to $W[i, j, \perp]$ contains at least one vertex in $\{U[i, j, 1], U[i, j, 2], U[i, j, 3]\}$.
As $(\{V[\perp, \perp, \perp]\} \cup W) \subseteq X_{\el}$,  and $G[X_{\el}]$ is connected,   $Y_s$ contains at least one vertex in the set.
As this is true for every $i \in [q]$ and $j \in [\ell]$, we have $|Y_s| \ge q \cdot \ell$.
By the third condition mentioned in \autoref{lemma:solution-edge-pair},  $|Y_s| - |X_s| \le k - d$.
Substituting $k - d = (\ell - 1) \cdot q$, we get $q \le |X_s|$.

As $X_s \cap (\{V[\perp, \perp, \perp]\} \cup W) = \emptyset$ and $X := V \cup W \cup \{V[\perp, \perp, \perp] \}$,  we have  $X_s \subseteq V$.
By the first condition mentioned in \autoref{lemma:solution-edge-pair},  $(X \setminus X_s) \cup Y_s$ is a vertex cover of $G'$.
As each $V_i$ is a clique in $G'$, {we have} that $|X_s \cap V_i| \le 1$ for every $i \in [q]$.
This, along with the fact {that} $q \le |X_s|$ {imply} that $|X_s \cap V_i| = 1$ for every $i \in [q]$.
Recall that $G'[V]$ is {isomorphic to} $G$.
As $G'[X_s]$ is an independent set in $G'$, {it follows} that $X_s$ is also an independent set in $G$.
It is also evident that it is multicolored.
This implies that $(G, q, \langle V_1, V_2, \dots, V_q \rangle)$ is a \yes-instance of $(3 \times q)$-\textsc{Multicolored Independent Set}.
\end{proof}

%\igm{Can we indeed assume that the input of \autoref{thm:np-hard} is a bipartite graph with a bipartition $\langle X, Y \rangle$ such that $X$ is a minimum vertex cover of $G$, as we claim in the introduction? If yes, we need to add this in the statement, and argue it in the proof}

As mentioned before, the first and the second point in the statement of the theorem {follow} directly from \autoref{obs:bound-k-rank} and \autoref{obs:bound-d-k}, respectively.
\autoref{lemma:np-hard-forward} and \autoref{lemma:np-hard-backward} imply that the reduction is correct.
By the description of the reduction, it outputs the constructed instance in polynomial time.
Hence, the third point in the statement of \autoref{thm:np-hard} is correct, which concludes its proof.

\section{\W[1]-\Hard ness results}
\label{sec:w-hard}

In this section we prove \autoref{thm:w1-hard}. That is, we show that \textsc{Contraction(\vc)} is \W[1]-\Hard\ when parameterized by the solution size $k$ plus the measure $d$.
Moreover, unless {the} \ETH\ fails, it does not admit an algorithm running in time $f(k + d) \cdot n^{o(k + d)}$ for any computable function {$f: \mathbb{N} \mapsto \mathbb{N}$}.
To obtain these results, we introduce {the} \textsc{Edge Induced Forest} problem.
We define this problem formally in \autoref{subsec:edge-induced-forest}, and present a parameter preserving reduction from \textsc{Multicolored Independent Set} to it.
This reduction, along with known results about \textsc{Multicolored Independent Set}, imply the corresponding result for \textsc{Edge Induced Forest}.
The proof is presented in \autoref{thm:w-hardness-edge-induced-forest}.
In \autoref{subsec:w1-contraction-vc}, we present a parameter preserving reduction from \textsc{Edge Induced Forest} to \textsc{Contraction(\vc)}.
This reduction, along with \autoref{thm:w-hardness-edge-induced-forest}, imply the correctness of \autoref{thm:w1-hard}.

\subsection{Edge Induced Forest is \W[1]-\Hard}
\label{subsec:edge-induced-forest}

We define the following problem.

\defproblem{\textsc{Edge Induced Forest}}{{A graph} $G$ and an integer $\ell$.}{Does there exist a set $F$ of at least $\ell$ edges in $G$ such that $G[V(F)]$ is a forest?}

We note that a similar problem called \textsc{Induced Forest} has already been studied.
In this problem, the input is the same but the objective is to find a subset $X$ of \emph{vertices} of $G$ of size at least $\ell$ such that $G[X]$ is a forest.
The general result of Khot and Raman \cite{KhotR02} implies that \textsc{Induced Forest} is \W[1]-\Hard\ when parameterized by the size of the solution $\ell$.
As expected, we can prove a similar result for \textsc{Edge Induced Forest}.

\begin{theorem}
\label{thm:w-hardness-edge-induced-forest}
{\sc Edge Induced Forest}, parameterized by the size of the solution $\ell$,  is \emph{\W[1]-\Hard}.
Moreover, unless {the} \emph{\ETH} fails, it does not admit an algorithm running in time $f(\ell) \cdot n^{o(\ell)}$ for any computable function {$f: \mathbb{N} \mapsto \mathbb{N}$}.
\end{theorem}
\begin{proof}
We present a simple parameter preserving reduction from \textsc{Multicolored Independent Set}.
The reduction takes as input an instance $(G, q, \langle V_1, V_2, \dots, V_q \rangle)$ of \textsc{Multicolored Independent Set}, and constructs another graph $G'$ from $G$ by adding a universal vertex $\alpha$ to $G$.
Formally, it adds a vertex $\alpha$ to $V(G)$, and adds edge $u\alpha$ to $E(G)$ for every vertex {$u$} in $V(G) \setminus \{\alpha\}$ to obtain $G'$.
It adds $q + 1$ pendant vertices adjacent to $\alpha$.
Formally,  for every $i \in [q + 1]$, it adds a vertex $x_i$ to $V(G')$, and an edge $x_i\alpha$ to $E(G')$.
Let $P$ be the collection of all the pendant vertices added in this step.
It sets $\ell = 2 \cdot q + 1$, and returns the instance $(G', \ell)$ of \textsc{Edge Induced Forest} as the constructed instance.
This completes the description of the reduction.

We now argue that the reduction is safe.
In the forward direction, suppose that $Q$ is a multicolored independent set in $G$.
Define $F := \{x_i \alpha\ |\ \forall\ x_i\in P\} \cup \{u_i \alpha\ |\ u_i \in Q \cap V_i\ \forall\ i \in [q] \}$.
It is easy to verify that $F$ is a solution of $(G', \ell)$.

In the reverse direction, suppose that $F$ is a solution of $(G', \ell)$, i.e., $G'[V(F)]$ is a forest and $|F| \ge \ell$.
We first argue that $\alpha \in V(F)$.
Assume, for the sake of contradiction, that $\alpha$ is not in $V(F)$.
This implies {that} $F$ contains $2 \cdot q + 1$ many edges in $G' - \{\alpha\}$.
Note that every vertex in $P$ is an isolated vertex in $G' - \{\alpha\}$.
Hence,  $V(F) \subseteq \bigcup_{i \in [q]} V_i$.
As $|F| \le V(F)$, {it follows that} there exists $i \in [q]$ such that $|V(F) \cap V_i| \ge 3$.
However, as $V_i$ is a clique in $G'$, this contradicts the fact that $G'[V(F)]$ is a forest.
Hence, $\alpha \in V(F)$.

As $\alpha \in V(F)$, it is safe to assume that $F$ contains all the edges in $F_P := E_{G'}(\{\alpha\}, P)$.
Since $|P| = q + 1$, $F$ contains at least $q$ many edges whose endpoints are in $V(G') \setminus P$.
The following two statements are direct consequences of the facts that $\alpha$ is an universal vertex, $V_i$ is a clique in $G'$ for any $i \in [q]$, and $G'[F]$ is a forest:
$(i)$ For any $i \in [q]$, we have $|V(F) \cap V_i| \le 1$, and in particular $|F \cap E_{G'}(V_i)| = 0$.
$(ii)$ For any $i \neq j \in [q]$, $|F \cap E_{G'}(V_{i}, V_{j})| = 0$.
This implies that every edge in $F \setminus F_P$ has $\alpha$ as one of its endpoints.
As there are at least $q$ edges $F \setminus F_P$, $|V(F) \cap V_i| = 1$ for every $i \neq j \in [q]$.

We define a subset $Q$ of $V(G)$ as
$Q := \{u \in V(G)\ |\ u\alpha \in F \text{ and } u \in V_i \text{ for some } i \in [q]\}.$
As for every $i \in [q]$, we have $|V(F) \cap V_i| = 1$, this implies $|Q \cap V_i| = 1$.
We argue that $Q$ is a multicolored independent set in $G$.
Consider any two indices $i \neq j \in [q]$, and let $u_i, u_j$ be the unique vertices in $Q \cap V_i$ and $Q \cap V_j$, respectively.
If $u_iu_j \in E(G)$ then $u_iu_j \in E(G')$ as $E(G) \subseteq E(G')$.
However,  as $u_i\alpha, u_j \alpha \in F$, this contradicts the fact that $G'[F]$ is a forest.
Hence,  vertices $u_i$ and $u_j$ are not adjacent in $G$.
Since $i, j$ are arbitrary indices in $[q]$, this is true for any $i \neq j \in [q]$, and therefore $Q$ is a multicolored independent set in $G$.

This implies that $(G, q, \langle V_1, V_2, \dots, V_q \rangle)$ is a \yes-instance of \textsc{Multicolored Independent Set} if and only if $(G', \ell)$ {is a \yes-instance} of \textsc{Edge Induced Forest}.
By the description of the reduction, it outputs the constructed instance in polynomial time.
The \W[1]-\Hardness\ of the problem follows from \autoref{prop:ind-set-w-hard}.
It is also easy to see that if \textsc{Edge Induced Forest} admits an algorithm with running time $f(\ell) \cdot n^{o(\ell)}$ {for some computable function $f: \mathbb{N} \mapsto \mathbb{N}$}, then \textsc{Multicolored Independent Set} also admits an algorithm with running time $f(q) \cdot n^{o(q)}$, which contradicts \autoref{prop:ind-set-w-hard}.
\end{proof}

\subsection{Contraction(\vc) is \W[1]-\Hard}
\label{subsec:w1-contraction-vc}

In this subsection we present a parameter preserving reduction from \textsc{Edge Induced Forest} to \textsc{Contraction(\vc)}.

\medskip
\noindent \textbf{The reduction:}
The reduction takes as input an instance $(G, \ell)$ of  \textsc{Edge Induced Forest} and returns an instance $(G', k, d)$ of \textsc{Contraction(\vc)}.
It constructs a graph $G'$ from $G$ as follows.
\begin{itemize}
\item It initializes $V(G') = E(G') = \emptyset$.
\item For every vertex $u$ in $V(G)$, it adds two vertices $z_u, p_u$ to $V(G')$ and the edge $z_up_u$ to $E(G')$.
\item For every edge $uv$ in $E(G)$, it adds {the vertex set} $\{y^{a}_{uv}$,  $y^{b}_{uv}$,  $y^{c}_{uv}$,  $w^{1}_{uv}$,  $w^2_{uv}$,  $p^{1}_{uv}$,  $p^2_{uv}\}$ to $V(G')$.
It adds edge $\{z_u y^{c}_{uv},  z_v y^{c}_{uv}\}$ to $E(G')$.
These edges encode adjacency relations in $G$.
It also adds edges $\{y^{a}_{uv}y^{b}_{uv}$,  $y^a_{uv}y^c_{uv}$,  $y^{b}_{uv}w^1_{uv}$,  $y^{b}_{uv}w^2_{uv}$,  $w^1_{uv} p^1_{uv}$,  $w^2_{uv} p^2_{uv} \}$ to $E(G')$.
These edges are part of a gadget which is private to edge $uv$.
\end{itemize}
This completes the construction of $G'$.
The reduction sets $k = 4 \cdot \ell$, $d = 3 \cdot \ell$, and returns $(G', k, d)$ as the constructed instance.
This completes the description of the reduction. Note that, indeed, $k < \rank(G')$ and $d \le k < 2d$ (more precisely, $k-d = \frac{d}{3}$).
See \autoref{fig:w-hardness-cont-vc} for an illustration.%\igm{the labels of the vertices in lower parts of the figure are not aligned}

\begin{figure}[t]
  \begin{center}
    \includegraphics[scale=0.65]{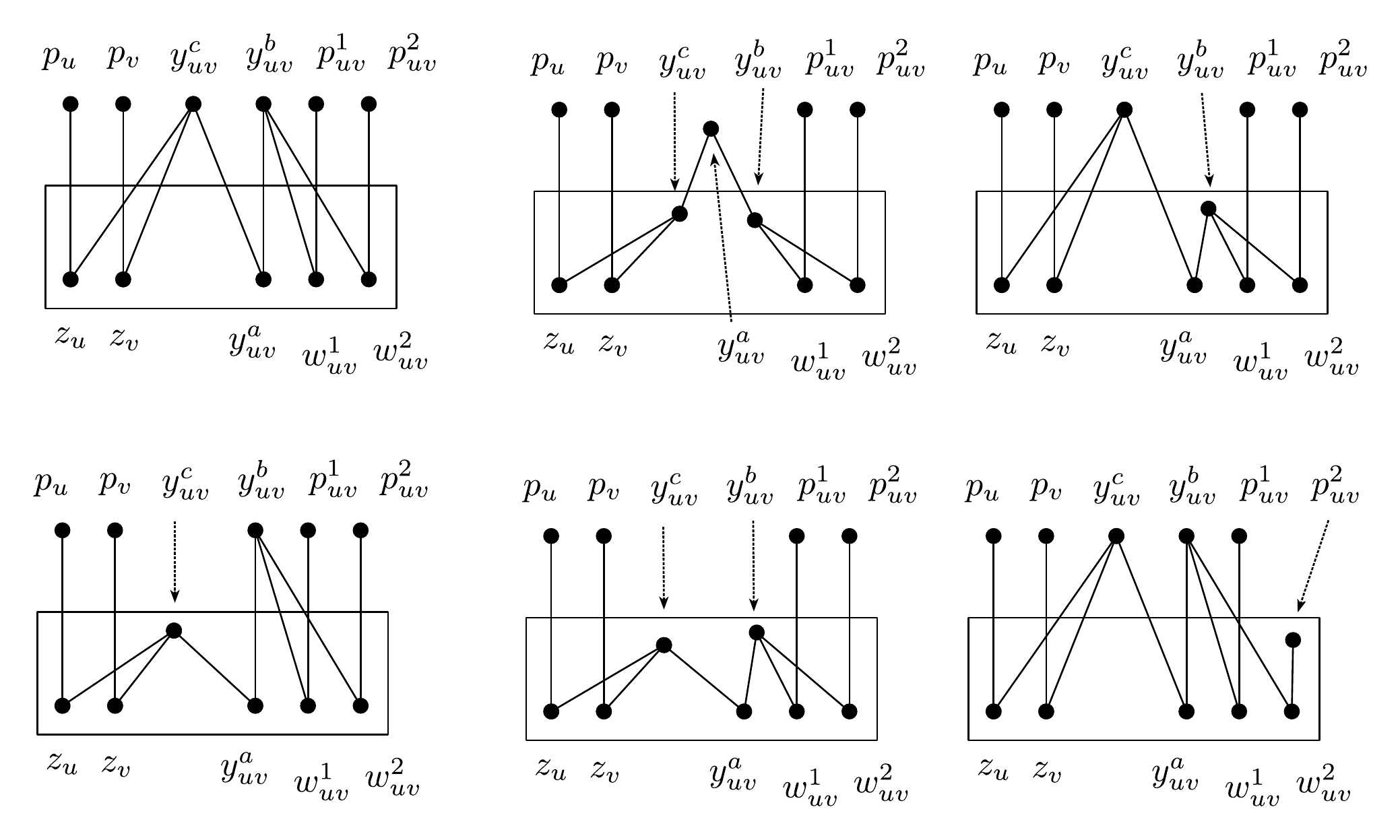}
    \end{center}
   \caption{{The} top-left figure illustrates an encoding of edge $uv$ in $G$ while reducing from an instance of \textsc{Edge Induced Forest} to an instance of \textsc{Contraction(\vc)}. The remaining five {figures correspond} to the partition of $Y_s$ mentioned in the proof of \autoref{lemma:w-hard-backward}. \label{fig:w-hardness-cont-vc}}
\end{figure}

Before proving the correctness of the reduction, we first note some properties of the graph $G'$.
We define the following sets:
\begin{itemize}
\item $Z := \{z_u \in V(G')\ |\ u \in V(G)\}$,
\item $Y := Y^a \cup Y^b \cup Y^c$ where $Y^{a} := \{y^{a}_{uv} \in V(G')\ |\ uv \in E(G)\}$,  $Y^{b} := \{y^{b}_{uv} \in V(G')\ |\ uv \in E(G)\}$, and $Y^{c} := \{y^{c}_{uv} \in V(G')\ |\ uv \in E(G)\}$,
\item $W := \{w^1_{uv}, w^2_{uv}\ |\ uv \in E(G)\}${, and}
\item $P := \{p_u \in V(G)\ |\ u \in V(G)\} \cup \{p^1_{uv}, p^2_{uv}\ |\ uv \in E(G)\}$.
\end{itemize}
Note that $\langle Z, Y, W, P \rangle$ is a partition of $V(G')$, each vertex in $P$ is a pendant vertex, and each vertex in $Z \cup W$ is adjacent to a pendant vertex in $P$.
Moreover, $\rank(G') > k$.
%In the next few paragraphs, we present an overview of the reduction and the proof of correctness of its backward direction.
%The reduction starts with the incidence graph of $G$.
%It adds pendant vertices to all vertices in $V(G')$ that encode a vertex in $V(G)$.
%Finally, it adds a private gadget for each vertex in $V(G')$ that encodes an edge in $E(G)$.
%For example, in \autoref{fig:w-hardness-cont-vc} (top left corner), the vertex $z_u$ corresponds to vertex $u$ in $V(G)$, and $p_u$ is a pendant vertex adjacent to it.
%Vertex $y^c_{uv}$ in $V(G')$ encodes the edge $uv$ in $G$, and vertices $\{y^a_{uv}, y^b_{uv}, w^1_{uv}, w^2_{uv}, p^1_{uv}, p^2_{uv}\}$ are part of the private gadget corresponding to $uv$.
Note that $X$ is an independent set in $G'$.
In the next lemma, we argue that it is also a minimum vertex cover of $G'$, which implies that, as claimed in the statement of \autoref{thm:w1-hard},  $G'$ is a bipartite graph with a bipartition $\langle X, Y \rangle$ such that $X$ is a minimum vertex cover of $G'$.

\begin{lemma}
\label{lemma:reduction-w-hard-min-vc}
The set $X := Z \cup W \cup Y^a$ is a minimum vertex cover of $G'$.
\end{lemma}
\begin{proof}
By the construction of $G'$,  it follows that $X$ is a vertex cover of $G'$.
To prove that it is a minimum vertex cover, we show that there is a matching of size {$|X|$} in $G$.
Consider the following set of edges $M := \{z_up_u\ |\ u \in V(G)\} \cup \{w^1_{uv} p^1_{uv}, w^2_{uv} p^2_{uv}\ |\ uv \in E(G) \} \cup \{y^{a}_{uv}y^{c}_{uv}\ |\ uv \in E(G)\}$.
It is easy to verify that $M$ is a matching in $G'$ of size $|X|$.
Hence, any vertex cover has size at least $|X|$.
This implies that $X$ is a minimum vertex cover of $G'$.
\end{proof}

\begin{lemma}
\label{lemma:w-hard-forward}
If $(G,\ell)$ is a {\sc Yes}-instance of {\sc Edge Induced Forest}, then $(G', k, d)$ is a {\sc Yes}-instance of {\sc Contraction(\vc)}.
\end{lemma}
\begin{proof}
{Let} $F$ be a solution of $(G, \ell)$ i.e., $G[V(F)]$ is a forest and $|F| \ge \ell$.
We assume, without loss of generality, that $|F| = \ell$.
By \autoref{lemma:reduction-w-hard-min-vc},  the set $X := Z \cup W \cup Y^a$ is a minimum vertex cover of $G'$.
Note that $X$ is also an independent set in $G'$.
We denote the independent set $V(G') \setminus X$ by $Y$.

We construct a solution pair $\langle X_s,  Y_s \rangle$ using $F$.
Define $X_s := \{y^a_{uv} \in Y\ |\ uv \in F\}$, and $Y_s := Y^b_s \cup Y^c_s$ where $Y^b_s := \{y^b_{uv} \in Y\ |\ uv \in F\}$, and $Y^c_s := \{y^c_{uv} \in X\ |\ uv \in F\}$.
It is easy to verify that $X_{\el} = (X \setminus X_s) \cup Y_s = Z \cup W \cup Y^b_s \cup Y^c_s$ is a vertex cover of $G'$ and its size is $|X| - \ell + 2 \cdot \ell = |X| + k - d$.
By the construction of $G'$, one can obtain graph $G'[Z \cup Y^c_s]$ by subdividing every edge in $G[F]$.
Hence, $G'[Z \cup Y^c_s]$ is a forest with $2\cdot \ell$ edges and some isolated vertices.
Also, $G'[Y^b_s \cup W]$ is a forest with at least $2 \cdot \ell$ edges (two edges corresponding to each vertex in $Y^b_s$).
As the vertices in $Z \cup Y^c$ and $Y^b \cup W$ are not adjacent, $G'[Z \cup Y^c_s \cup Y^b_s \cup W]$ is a forest with $4 \ell$ edges.
This implies that $\rank(X_{\el}) = 4 \cdot \ell$.
Hence, $\langle X_s, Y_s \rangle$ satisfies the three conditions mentioned in \autoref{lemma:solution-edge-pair}.
As $k < \rank(G)$, \autoref{lemma:solution-edge-pair} implies that there is a subset $F'$ of $E(G')$ such that $\vc(G'/F') \le \vc(G) - d$.
Hence, $(G', k, d)$ is a \yes-instance of \textsc{Contraction(\vc)}.
\end{proof}

We first present a brief overview of the proof of the correctness in the backward direction.
By \autoref{lemma:solution-edge-pair}, there is a solution $F$ of $(G', k, d)$ if and only if there exists a solution pair $\langle X_{s}, Y_s \rangle$ such that
$(i)$ $X_{\el} = (X \setminus X_{s}) \cup Y_s$ is a vertex cover of $G'$,
$(ii)$ $\rank(X_{\el}) \ge |F| = k = 4 \cdot \ell$, and
$(iii)$ $|Y_s| - |X_s| \le k - d = \ell$.
Note that as $X$ and $Y = V(G) \setminus X$ are independent sets in $G'$, every edge in $E(X_{\el})$ is incident on exactly one vertex in $Y_s$.
We can interpret the second condition as a \emph{value function} and the third condition as a \emph{cost function}.
In other words, our objective is to find sets $X_s, Y_s$ such that their cost, i.e., $|Y_s| - |X_s|$, is at most $\ell$ whereas their value, i.e., the rank of edges in $E(X_{\el})$ that are incident on $Y_s$, is at least $4 \cdot \ell$.
\autoref{lemma:pendant-out-of-pair} implies that the vertices of the form $z_u, w^1_{uv}$, and $w^{2}_{uv}$ are in $X_{\el}$.
The first condition implies that only the five configurations shown in \autoref{fig:w-hardness-cont-vc} are possible (the top-left is \emph{not} a configuration).
Starting from top-middle and moving row-wise, the individual value and cost of these configurations are $(4, 1)$, $(3, 1)$, $(3, 1)$, $(6, 2)$, and $(1, 1)$, respectively.
To meet both the value and budget constraints, every vertex in $X_s, Y_s$ needs to be the of first type.
This implies there are $\ell$ vertices in $X_s$ that are of the form $y^{a}_{uv}$, and $Y_s$ contains the corresponding vertices of the form $y^{b}_{uv}$ and $y^{c}_{uv}$.
We argue that the edges corresponding to vertices in $Y^{c}_{uv}$ form a solution of $(G, \ell)$ and formalize these ideas in the next lemma.

\begin{lemma}
\label{lemma:w-hard-backward}
If $(G', k, d)$ is a {\sc Yes}-instance of {\sc Contraction(\vc)}, then $(G,\ell)$ is a {\sc Yes}-instance of {\sc Edge Induced Forest}.
\end{lemma}
\begin{proof}
Suppose that $F'$ is a solution of $(G', k, d)$, i.e., $\vc(G'/F') \le \vc(G') - d$ and $|F'| \le k$.
As $k < \rank(G)$, we can assume, without loss of generality, that $|F'| = k$.
\autoref{lemma:solution-edge-pair} implies that there exists a solution pair $\langle X_s, Y_s \rangle$ that satisfies the three conditions mentioned in its statement.
Recall that every vertex in $Z \cup W$ is adjacent to some pendant vertex in $G'$.
\autoref{lemma:pendant-out-of-pair} implies that there exists a solution pair $\langle X_s, Y_s \rangle$ with {the} additional property that $X_s \cap (Z \cup W) = \emptyset$.
As $X_s \subseteq X = Z \cup W \cup Y^a$, this implies that $X_s \subseteq Y^a$.

We argue that $|X_s| = \ell$.
%Since, the solution pair $\langle X_s, Y_s \rangle$ satisfies the first condition, i.e., $(X \setminus X_s) \cup Y_s$ is a vertex cover of $G'$, we have $N(X_s) \subseteq Y_s$.
%We argue that $N(X_s) = Y_s$.
We partition the vertices in $Y_s$ {into} the following five sets.
\begin{itemize}
\item $Y[1, 1, 1] := \{y^b_{uv}, y^c_{uv} \in Y_s\ |\ (y^a_{uv} \in X_s)\ \land\ (y^b_{uv} \in Y_s)\ \land\ (y^c_{uv} \in Y_s) \}$.
\item $Y[0, 1, 0] := \{y^b_{uv} \in Y_s\ |\ (y^a_{uv} \not\in X_s)\ \land\ (y^b_{uv} \in Y_s)\ \land\ (y^c_{uv} \not\in Y_s) \}$.
\item $Y[0, 0, 1] := \{y^c_{uv} \in Y_s\ |\ (y^a_{uv} \not\in X_s)\ \land\ (y^b_{uv} \not\in Y_s)\ \land\ (y^c_{uv} \in Y_s) \}$.
\item $Y[0, 1, 1] := \{y^b_{uv}, y^c_{uv} \in Y_s\ |\ (y^a_{uv} \not\in X_s)\ \land\ (y^b_{uv} \in Y_s)\ \land\ (y^c_{uv} \in Y_s) \}$.
\item $Y[0, 0, 0] := Y_s \cap P$.
\end{itemize}

We can define the  sets $Y[1, 0, 0]$, $Y[1, 0, 1]$, and $Y[1, 1, 0]$ in a similar way.
Note that $y^a_{uv} \in X_s$ implies that $y^b_{uv}, y^c_{uv} \in Y_s$.
Hence, we do not need to consider the sets $Y[1, 0, 0]$, $Y[1, 0, 1]$, and $Y[1, 1, 0]$.
This also implies $2 \cdot |X_s| = |Y[1, 1, 1]|$.
Hence, to argue that $|X_s| = \ell$,  it is sufficient to prove that $|Y[1, 1, 1]| = 2 \cdot \ell$.
As the solution pair $\langle X_s, Y_s \rangle$ satisfies the third condition, i.e., $|Y_s| - |X_s| \le k - d$, we have
$$|Y[1, 1, 1]| + |Y[0, 1, 0]| + |Y[0,0,1]| + |Y[0,1,1]| + |Y[0,0,0]| - |X_S| \le \ell.$$
Substituting $|Y[1, 1, 1,]| = 2 \cdot |X_s|$, and multiplying by two, we get the following relation.
\begin{equation}
\label{eq:size-upper-bound}
|Y[1, 1, 1]| + 2 \cdot |Y[0, 1, 0]| + 2 \cdot |Y[0,0,1]| + 2 \cdot |Y[0,1,1]| + 2 \cdot |Y[0,0,0]| \le 2 \cdot \ell.
\end{equation}
Define $X_{\el} := (X \setminus X_s) \cup Y_s$.
By the definition, $|E(X_{\el})| \ge \rank(X_{\el})$.
As the solution pair $\langle X_s, Y_s \rangle$ satisfies the second condition, i.e., $\rank((X \setminus X_s) \cup Y_s) \ge 4 \cdot \ell$, we have $|E(X_{\el})| \ge \rank(X_{\el}) \ge 4 \cdot \ell$.
As $X, Y$ both are independent sets in $G'$, every edge in $E(X_{\el})$ {has} one of its {endpoints} in $X \setminus X_s$ and {the other} one in $Y_s$.
It is easy to verify {(}see \autoref{fig:w-hardness-cont-vc}{)} that the number of edges incident on each vertex in $Y[1, 1, 1]$, $Y[0, 1, 0]$, $Y[0,0,1]$, $Y[0,1,1]$, $Y[0,0,0]$ {is} $2$, $3$, $3$, $3$, and $1$, respectively.
Substituting these values we get
$$2 \cdot |Y[1, 1, 1]| + 3 \cdot |Y[0, 1, 0]| + 3 \cdot |Y[0,0,1]| + 3 \cdot |Y[0,1,1]| + |Y[0,0,0]| \ge 4 \cdot \ell.$$
Dividing the inequality by two yields the following relation.
\begin{equation}
\label{eq:rank-lower-bound}
|Y[1, 1, 1]| + \frac{3}{2} \cdot |Y[0, 1, 0]| + \frac{3}{2} \cdot |Y[0,0,1]| + \frac{3}{2} \cdot |Y[0,1,1]| + \frac{1}{2} \cdot |Y[0,0,0]| \ge 2 \cdot \ell.
\end{equation}
Equation~(\ref{eq:size-upper-bound}) and (\ref{eq:rank-lower-bound}) imply that the only feasible case is when $|Y[1, 1, 1]| = 2 \cdot \ell$ and all other sets {have} cardinality zero.
This implies $|X_s| = \ell$.
Also, $|E(X_{\el})| = 2 \cdot |Y[1, 1, 1]| = 4 \cdot \ell$.
As $\rank(X_{\el}) \ge 4 \cdot \ell$, {it follows} that $G'[X_{\el}]$ is a forest with $4 \cdot \ell$ edges.
It is easy to verify that the graph induced on $Z \cap N[Y_s \cap Y^c]$ is a forest with $2 \cdot \ell$ edges.

We now construct a solution of $(G, \ell)$ using the set $X_s$,  more precisely $Y_s \cap Y^c$.
Define $F := \{uv \in E(G)\ |\ y^c_{uv} \in Y_s\}$.
As $|X_s| = \ell$, we have $|F| = \ell$.
It remains to argue that $G[V(F)]$ is a forest.
By the construction of $G'$, one can obtain $G'[Z \cap N[Y \cap Y^c]]$ by subdividing every edge in $G[V(F)]$.
As the former graph is a forest, we can conclude that $G[V(F)]$ is also a forest.
Hence, $F$ is a solution of $(G, \ell)$.
This implies that if $(G', k, d)$ is a {\yes-instance} of \textsc{Contraction(\vc)} then $(G, \ell)$ is a {\yes-instance} of \textsc{Edge Induced Forest}.
\end{proof}

We are ready to present the proof of \autoref{thm:w1-hard}.

\begin{proofThm2}% (of \autoref{thm:w1-hard})
Consider the reduction presented in this subsection.
\autoref{lemma:w-hard-forward} and \autoref{lemma:w-hard-backward} imply that the reduction is safe.
By the description of the reduction, it outputs the constructed instance in polynomial time.
The \W[1]-\Hardness\ of \textsc{Contraction(\vc)} follows from \autoref{thm:w-hardness-edge-induced-forest}.
As $k = 4 \cdot \ell$ and $d = 3 \cdot \ell$, if \textsc{Contraction(\vc)} admits an algorithm with running time $f(k + d) \cdot n^{o(k + d)}$, then \textsc{Edge Induced Forest} also admits an algorithm with running time $f(\ell) \cdot n^{o(\ell)}$, which contradicts \autoref{thm:w-hardness-edge-induced-forest}.
\end{proofThm2}

\section{Algorithm for \textsc{Contraction(\vc)}}
\label{sec:algorithm}

In this section we prove \autoref{thm:algorithm}.
We present an algorithm that takes as input an instance $(G, k, d)$ of \textsc{Contraction(\vc)}, and returns either \yes\ or \no,
%The algorithm considers the three exhaustive cases and treats them separately.
whose high-level description is as follows (cf. \autoref{fig:diagram-algo}):

\begin{itemize}
\item If $k = \rank(G)$, then it uses the algorithm mentioned in \autoref{lemma:k-equal-rank-algo}.
\item If $k < \rank(G)$ and $2d \le k$, then it uses the algorithm mentioned in \autoref{lemma:k-larger-2k-algo}.
\item If $k < \rank(G)$ and $d \le k < 2d$, then it uses the algorithm mentioned in \autoref{lemma:k-lesser-2k-algo}.
\end{itemize}

Note that, since we can safely assume that $d \le k \leq \rank(G)$, the above three cases are exhaustive.
We handle each of {these cases} in the next three subsections (note that the first two are much easier than the last one). %respectively %\igm{even if the first two subsections are short and similar, I think it is more structured to leave it like it is}.
\autoref{sub-sec:entire-algo} contains the correctness proof of \autoref{thm:algorithm}.
Throughout this section, we assume that $G$ is a connected graph.
We justify this assumption in \autoref{sub-sec:entire-algo}.

\subsection{First case: $k = \rank(G)$}
\label{sec:algo-first-case}

It is sufficient to prove the following lemma to handle this case.

\begin{lemma}
\label{lemma:k-equal-rank-algo}
There exists an algorithm that, given as input an instance $(G, k, d)$ of {\sc Contraction(\vc)} with a guarantee that $k = \emph{\rank}(G)$, runs in time $1.2738^{d}\cdot n^{\calO(1)}$, and correctly determines whether it is a {\sc Yes}-instance.
\end{lemma}
\begin{proof}
Consider an algorithm that for input $(G, k, d)$, runs the algorithm mentioned in \autoref{prop:find-vc} as a subroutine with $G$ and $d - 1$ as its input.
If the subroutine concludes that $\vc(G) \le d - 1$, then the algorithm returns \no, otherwise it returns \yes.
This concludes the description of the algorithm. Its correctness and running time follow from \autoref{obs:bc-rank-equiv} and \autoref{prop:find-vc}, respectively.
\end{proof}

\subsection{Second case: $k < \rank(G)$ and $2d \le k$}
\label{sec:algo-second-case}

As in the previous subsection, it is sufficient to prove the following lemma.

\begin{lemma}
\label{lemma:k-larger-2k-algo}
There exists an algorithm that, given as input an instance $(G, k, d)$ of {\sc Contraction(\vc)} with guarantees that $k < \emph{\rank}(G)$ and $2d \le k$, runs in time $1.2738^{d} \cdot n^{\calO(1)}$, and correctly determines whether it is a {\sc Yes}-instance.
\end{lemma}
\begin{proof}
Consider an algorithm that for input $(G, k, d)$, runs the algorithm mentioned in \autoref{prop:find-vc} as a subroutine with $G$ and $d$ as its input.
It considers the following three cases depending on the value of $\vc(G)$.
$Case~(i)\ (\vc(G) < d)$:
It concludes that $(G, k, d)$ is a \no-instance.
$Case~(ii)\ (\vc(G) = d)$:
It concludes that $(G, k, d)$ is a \no-instance.
$Case~(iii)\ (\vc(G) > d)$:
It concludes that $(G, k, d)$ is a \yes-instance.
This completes the description of the algorithm.

We now argue the correctness of the algorithm.
As the vertex cover number of any graph is a non-negative integer, if $\vc(G) < d$ then the input is a \no-instance.
Note that to eliminate all edges in a {connected} graph by contracting edges, one needs to contract all the edges in a spanning tree.
Hence,  if $\vc(G) = d$ then the only feasible solution of $(G, k, d)$ is a spanning tree of $G$.
However, as $k < \rank(G)$,  the algorithm correctly concludes that it is a \no-instance.
For the third case, consider a subroutine that finds two vertices in a minimum vertex cover that are at distance at most two and contracts a shortest path between these two vertices.
The existence of such vertices is guaranteed by the fact that $G$ is a connected graph.
Note that this path is of length one or two.
In each iteration of the process, $k$ drops by at most two and $\vc(G)$ drops by one.
As $2d \le k$, if $\vc(G) > d$ then the subroutine can repeat the process $d$ times.
Hence, the algorithm correctly concludes that the input instance is a \yes-instance in the third step.

The running time of the algorithm follows from its description and  \autoref{prop:find-vc}.
\end{proof}

\subsection{Third case: $k < \rank(G)$ and $d \le k < 2d$}
\label{sec:algo-third-case}

The objective of this subsection is to prove the following lemma.
%\red{Note that as $d \le k < 2d$, we have $\min\{k - d, d\} = k - d$.} \igm{Do we really need to say this? (see my remark in the abstract)}

\begin{lemma}
\label{lemma:k-lesser-2k-algo}
There exists an algorithm that, given as input an instance $(G, k, d)$ of {\sc Contraction(\vc)} with guarantees that $k < \emph{\rank}(G)$ and $d \le k < 2d$, runs in time $2^{\calO(d)} \cdot n^{k - d + \calO(1)}$, and correctly determines whether it is a {\sc Yes}-instance.
\end{lemma}

\noindent We refer readers to \autoref{sec:intro}, in particular \autoref{fig:diagram-algo}, for an overview of the algorithm presented in this subsection.

\subsubsection{Simplifying an instance of \textsc{Contraction(\vc)}}
\label{sub-sub-sec:simplify-instance}

We prove the following lemma, which will allow us to assume henceforth that we are equipped with a minimum vertex cover of the input graph with small rank.

\begin{lemma}
\label{lemma:large-oct-large-rank-vc}
There exists an algorithm that, given as input an instance $(G, k, d)$ of {\sc Contraction(\vc)} with guarantees that $k < \emph{\rank}(G)$ and $d \le k < 2d$, runs in time $2.6181^{k} \cdot n^{\calO(1)}$, and either correctly concludes that $(G, k,d)$ is a {\sc Yes}-instance, or computes a minimum vertex cover $X$ of $G$ such that $\emph{\rank(X)} < d$.
\end{lemma}
\begin{proof}
Consider an algorithm that, given $(G, k, d)$ as input, runs the algorithm mentioned in \autoref{prop:find-oct} as a subroutine with $G$ and $k$ as its input.
If $\oct(G) > k$, then it concludes that $(G, k, d)$ is a \yes-instance.
If $\oct(G) \le k$, then it uses the algorithm mentioned in \autoref{prop:find-vc-para-oct} to compute a minimum vertex cover $X$ of $G$.
If $\rank(X) \ge d$, then it concludes that $(G, k, d)$ is a \yes-instance.
Otherwise, it returns $X$ as the desired vertex cover.
This completes the description of the algorithm.

We argue the correctness of the algorithm.
Consider the case where $\oct(G) > k$.
Recall that we denote by $\bc(G)$ the minimum number of edges in $G$ that need to be contracted to make it a bipartite graph.
By \autoref{obs:oct-bound-bc}, $\oct(G) > k$ implies that $\bc(G) > k$.
Hence, by \autoref{obs:bc-rank-equiv}, for any partition $(V_L, V_R)$ of $V(G)$, we have $\rank(V_L) + \rank(V_R) > k$.
Consider a partition $(V_L, V_R)$ of $V(G)$ such that $V_L$ is a minimum vertex cover of $G$.
As $V_R$ is an independent set, $\rank(V_R) = 0$.
This implies $\rank(V_L) > k$.
Hence, we can reduce $\vc(G)$ by $d$ by contracting $d$ (which is at most $k$) edges whose both endpoints are in $V_L$.
Consider the case when the algorithm finds a minimum vertex cover $X$ of $G$ such that $\rank(X) \ge d$.
Once again, we can reduce $\vc(G)$ by $d$ by contracting $d$ edges of a spanning forest of $G[X]$.
Hence, in both these cases, the algorithm correctly concludes that the input is a \yes-instance.
Otherwise, the algorithm returns a minimum vertex cover $X$ of $G$ such that $\rank(X) < d$.

The running time of the algorithm follows from its description and \autoref{prop:find-oct}.
\end{proof}

\subsubsection{Reducing to \textsc{Annotated Contraction(\vc)}}
\label{sub-sec:annotated-contr}

An input of the \textsc{Annotated Contraction(\vc)} problem consists of an instance $(G, k,d)$ of \textsc{Contraction(\vc)}, a minimum vertex cover $X$ of $G$, and two disjoint subsets $X_L, X_R$ of $X$.
We are interested in a vertex cover $X_{\el}$ of $G$ whose size is not much larger than that of $X$ but has rank at least $k$.
To construct $X_{\el}$ from $X$, we need to find a \emph{solution pair} $\langle X_s, Y_s\rangle$ such that vertices in $X_s$ are `moved out' of $X$, and vertices in $Y_s$ are `moved in'.
Given $\langle X_L, X_R \rangle$, we add a restriction on a possible solution pair $\langle X_s, Y_s\rangle$.
Namely, we are interested in $X_s$ that contains $X_R$ and is disjoint from $X_L$.
The following is the formal definition of the problem.

\defproblem{\textsc{Annotated Contraction(\vc)}}{An instance $(G, k, d)$ of \textsc{Contraction(\vc)},  a minimum vertex cover $X$ of $G$, and a tuple $\langle X_L, X_R \rangle$ such that $X_L, X_R $ are disjoint subsets of $X$.}{{Do} there exist sets $X_s \subseteq X$ and $Y_s \subseteq Y\ (= V(G) \setminus X)$ such that
$(i)$ $(X \setminus X_s) \cup Y_s$ is a vertex cover of $G$,
$(ii)$ $\rank((X \setminus X_s) \cup Y_s) \ge k$,
$(iii)$ $ |Y_s| - |X_s|  \le k - d$, and
$(iv)$ $X_L \cap  X_s = \emptyset$ and $X_R \subseteq X_s$?}

%\igm{About the input of \textsc{Annotated Contraction(\vc)}: In the 2nd line of this subsubsection we have said that $X_L,X_R$ are \textbf{disjoint}, but we don't write it here. Why? I see that it is implied by condition (iv) of the output, but we may add it in the input as well. Also, I don't think we need to define $Y$ in the statement of the problem, we can use $V(G) \setminus X$}

The first three conditions correspond to the three conditions mentioned in \autoref{lemma:solution-edge-pair}.
Given an instance $(G, k, d)$ of \textsc{Contraction(\vc)}, using \autoref{lemma:solution-edge-pair} we construct `\FPT-many' instances of \textsc{Annotated Contraction(\vc)} such that the original instance is a \yes-instance if and only if at least one of the newly created instances is a \yes-instance.
We remark that there is a small technical caveat while using \autoref{lemma:solution-edge-pair}.
Consider an instance $(G, k, d)$ of  \textsc{Contraction(\vc)}, and let $F$ be {a} solution.
\autoref{lemma:solution-edge-pair} implies that there are subsets $X_s \subseteq X$ and $Y_s \subseteq V(G) \setminus X$ such that
$(i)$ $(X \setminus X_s) \cup Y_s$ is a vertex cover of $G$,
$(ii)$ $\rank((X \setminus X_s) \cup Y_s) \ge |F|$, and
$(iii)$ $ |Y_s| - |X_s|  \le |F| - d$.
However, the statement of \textsc{Annotated Contraction(\vc)} specifies the integer $k$ and not the actual size of a minimum solution $F$.
For example, if there exists a solution $F$ of size, say, $k/2$, then \autoref{lemma:solution-edge-pair} ensures that $\rank((X \setminus X_s) \cup Y_s) \ge k/2$, however $\rank((X \setminus X_s) \cup Y_s)$ can be smaller than $k$.
To {overcome} this, we assume that $(G, k - 1, d)$ is a \no-instance of \textsc{Contraction(\vc)}.
This implies that if there is a subset $F$ of $E(G)$ of size \emph{at most} $k$ such that $\vc(G/F) \le \vc(G) - d$, then $F$ is of size \emph{exactly} $k$.
We summarize below all the assumptions on the input instance.
\begin{guarantee}
\label{guarantee:special-instance}
Consider an instance $(G, k, d)$ of {\sc Contraction(\vc)} that satisfies the following conditions.
\begin{itemize}
\item $G$ is a connected graph, $k < \emph{\rank}(G)$,  and $d \le k$.
\item A minimum vertex cover $X$ of $G$ is provided as an additional part of the input.
\item $\emph{\rank}(X) < d$.
\item $(G, k - 1, d)$ is a \no-instance of {\sc Contraction(\vc)}.
\end{itemize}
\end{guarantee}
Unless stated otherwise, we denote the independent set $V(G) \setminus X$ by $Y$.

%\igm{Up to here, we can assume that we are given a minimum vertex cover $X$ of the input graph, but not that it is bipartite, nor that $X$ is one of the sides of the bipartition}

Consider an instance $(G, k, d)$ of {\sc Contraction(\vc)} with \autoref{guarantee:special-instance}.
We construct $2^{\calO(d)}$ many instances of \textsc{Annotated Contraction(\vc)} such that $(G, k, d)$ is a \yes-instance if and only if at least one of these newly created instances is a \yes-instance.
Informally, let $F$ be the set of edges in a spanning forest of $G[X]$.
As $\rank(X) < d$, we have $|F| < d$.
We iterate over all `valid' partitions $\langle X_L, X_R \rangle$ of $V(F)$.
We construct an instance of \textsc{Annotated Contraction(\vc)} for each such {a} partition.
We formalize this intuition and prove its correctness in the following lemma.

\begin{lemma}
\label{lemma:contr-VC-to-annot-contr-vc}
Suppose that there is an algorithm that solves {\sc Annotated Contraction(\vc)} in time $f(n, k, d)$.
Then, there exists an algorithm that given as input an instance $(G, k, d)$ of {\sc Contraction(\vc)} with \autoref{guarantee:special-instance}, runs in time $ 3^d \cdot n^{\calO(1)} \cdot f(n, k, d)$, and correctly determines whether it is a \textsc{\sc Yes}-instance.
\end{lemma}
\begin{proof}
Let $\calA$ be an algorithm that, given an instance $((G, k, d), X, \langle X_L, X_R \rangle)$ of \textsc{Annotated Contraction(\vc)}, runs in time $f(n, k, d)$, and correctly determines whether it is a \yes-instance.
We describe an algorithm that solves \textsc{Contraction(\vc)} using $\calA$ as a subroutine.

The algorithm takes as input an instance $(G, k, d)$ of \textsc{Contraction(\vc)} and returns either \yes\ or \no.
By \autoref{guarantee:special-instance}, the input also consists of a minimum vertex cover $X$ of rank less than $d$.
Let $F_x$ {be the edge set of} a spanning forest of $G[X]$.
For every subset $F'_x$ of $F_x$, the algorithm constructs multiple instances of \textsc{Annotated Contraction(\vc)} as specified in the next paragraph.
The algorithm uses Algorithm~$\calA$ to check if at least one of these newly created instances is a \yes-instance.
If it is the case, then the algorithm returns \yes, otherwise it returns \no.

\sloppy Consider a subset $F'_x$ of $F_x$.
Let $\calP$ be the collection of partitions $\langle X_{L, F'_x}, X_{R, F'_x} \rangle$ of $V(F_x \setminus F'_x)$ such that for every edge $e$ in $F_x \setminus F'_x$, exactly one of its endpoints is in $X_{L, F'_x}$ and the other {one} is in $X_{R, F'_x}$.
For every partition $\langle X_{L, F'}, X_{R, F'} \rangle$ in $\calP$, the algorithm does as follows:
If $X_{R, F'_x}$ is \emph{not} an independent set in $G$, then the algorithm constructs a trivial \no-instance.
Otherwise, it adds $(G, k, d, X, \langle X_{L}, X_{R}\rangle)$ to the collection of instances of \textsc{Annotated Contraction(\vc)}.
Here, $X_L = X_{L, F'} \cup V(F\setminus F')$ and $X_R = X_{R, F'}$.
This completes the description of the algorithm.

We now argue the correctness of the algorithm.
Suppose that $(G, k, d)$ is a \yes-instance.
Recall that, by \autoref{guarantee:special-instance}, $(G, k - 1, d)$ is a \no-instance.
Hence, there exists a subset $F \subseteq E(G)$ of size exactly $k$ such that $\vc(G/F) \le \vc(G) - d$.
By \autoref{lemma:solution-edge-pair}, there are subsets $X_s \subseteq X$ and $Y_s \subseteq V(G) \setminus X$ such that
$(i)$ $(X \setminus X_s) \cup Y_s$ is a vertex cover of $G$,
$(ii)$ $\rank((X \setminus X_s) \cup Y_s) \ge |F|$, and
$(iii)$ $|Y_s| - |X_s|  \le |F| - d$.
As $|F| = k$, there is a solution pair $\langle X_s, Y_s\rangle$ that satisfies the first three conditions.
To see that the solution pair {also} satisfies the last condition mentioned in the problem statement, let $F'_x$ be the subset of $F_x$ such that $V(F \setminus F'_x) \cap (X \setminus X_s) = \emptyset$.
As $X_s$ is an independent set {in} $G$, for every edge $e$ in $F_x \setminus F'_x$, exactly one of its endpoints is in $X_s$.
As the algorithm constructs a new instance for every such a partition, at least one of the newly created instances is a \yes-instance.

The algorithm returns \yes\ only when Algorithm~$\calA$ returns \yes\ on one of the newly created instances.
By the correctness of Algorithm~$\calA$, at least one of the newly created instances is a \yes-instance.
Hence, there exists sets $X_s \subseteq X$ and $Y_s \subseteq Y$ such that $(i)$ $(X \setminus X_s) \cup Y_s$ is a vertex cover of $G$,
$(ii)$ $\rank((X \setminus X_s) \cup Y_s) \ge k$, and
$(iii)$ $ |Y_s| - |X_s|  \le k - d$.
By \autoref{lemma:solution-edge-pair}, there exists a subset $F \subseteq E(G)$ of size $k$ such that $\vc(G/F) \le \vc(G) - d$.
Hence, $(G, k, d)$ is a \yes-instance.
This concludes the proof of correctness of the algorithm.

For every $i \in \{0, 1, 2, \dots, d\}$, the algorithm iterates over all subsets of edges of size $i$.
It can construct the partition by guessing the right endpoint of the remaining $d - i$ edges.
For every partition, it creates an instance and executes Algorithm~$\calA$.
Hence, the total running time of the algorithm is $\calO(\sum_{i=0}^{d}\binom{d}{i} \cdot 2^{d - i} \cdot (f(n, d, k) + n^2)) = \calO(3^d \cdot (f(n, k ,d) + n^2))$.
This concludes the proof of the lemma.
\end{proof}

\label{comment:assumption-met}

As mentioned in the overview of the introduction, to solve an instance of  \textsc{Annotated Contraction(\vc)}, we reduce it to an equivalent instance of {the} \textsc{Constrained MaxCut} problem.
To present such a reduction, it is convenient to work with an instance $((G, k, d), X, \langle X_L, X_R \rangle)$ of  \textsc{Annotated Contraction(\vc)} where $X$ is an independent set.
We present a reduction rule that eliminates edges with both {endpoints} in $X$.
The reduction rule states that it is safe to contract edges with both endpoints  in $X_L$, and that it is safe to delete edges with one endpoint in $X_L$ and another endpoint in $X_R$.
Recall that if there is an edge with both endpoints in $X_R$, then the input is a trivial \no-instance.
Note that $\langle X_L, X_R \rangle$ is not a partition of $X$.
However, as we only need the following reduction rule for the instances obtained by the algorithm mentioned in \autoref{lemma:contr-VC-to-annot-contr-vc}, we can assume that $X \setminus (X_L \cup X_R)$ is an independent set in $G$.

\begin{reduction rule}
\label{rr:anno-contr-edges-X} Consider an instance $((G, k, d), X, \langle X_L, X_R \rangle)$ of  {\sc Annotated Contraction(\vc)}.
Let $F_1 = E(X_L, X_R)$ and $F_2$ be the set of all edges in a spanning forest of $G[X_L]$.
\begin{itemize}
\item Delete the edges in $F_1$.
\item Contract the edges in $F_2$ and reduce both $k, d$ by $|F_2|$.
\end{itemize}
Return the instance $((G', k', d'), X', \langle X'_L, X_R \rangle)$ where $G' = (G - F_1)/F_2$, $k' = k - |F_2|$, $d' = d - |F_2|$, $X' = V(G[X]/F_2)$, and $X'_L = V(G[X_L]/F_2)$.
\end{reduction rule}

\begin{lemma}
\label{lemma:rr-anno-contr-edges-X}
\autoref{rr:anno-contr-edges-X} is safe. Therefore, it is safe to assume that we are given an instance $((G, k, d), X, \langle X_L, X_R \rangle)$ of  {\sc Annotated Contraction(\vc)} such that $X$ is an independent set and a minimum vertex cover of $G$.
\end{lemma}
\begin{proof}
As $V(F_2) \subseteq X_L$, for any two subsets $X_s,Y_s$ such that $X_L \cap  X_s = X_L \cap Y_s = \emptyset$, we have $\rank((X \setminus X_s) \cup Y_s) \ge k$ if and only if $\rank((X' \setminus X_s) \cup Y_s) \ge k'$.
Here, $X' = V(G[X]/F_2)$.
Also, by the construction, $k - d = k' - d'$.

$(\Rightarrow)$ Suppose that $((G, k, d), X, \langle X_L, X_R \rangle)$ is a \yes-instance.
Then, there exists a solution pair $\langle X_s, Y_s\rangle$ that satisfies the four conditions mentioned in the definition of the problem.
We argue that $\langle X_s, Y_s\rangle$ is also a solution of $(G', k', d', X', \langle X'_L, X_R \rangle)$.
As $(X \setminus X_s) \cup Y_s$ is a vertex cover of $G$, it is also a vertex cover {of} $G - F_1$.
As $(X' \setminus X_s) \cup Y_s$ is obtained from $(X \setminus X_s) \cup Y_s$ by contracting the edges in $F_2$, whose both endpoints are in $X \setminus X_s$, the set $(X' \setminus X_s) \cup Y_s$ is a vertex cover of $(G - F_1)/F_2$.
As $V(F_2) \subseteq X_L$ and $X_L \cap  X_s = X_L \cap Y_s = \emptyset$, $\rank((X \setminus X_s) \cup Y_s) \ge k$ implies that $\rank((X' \setminus X_s) \cup Y_s) \ge k - |F_2| = k'$.
Also, by the construction, $k - d = k' - d'$.
Hence, $|Y_s| - |X_s| \le k' - d'$.
It is easy to verify that $X'_L \cap X_s = \emptyset$ and $X_R \subseteq X_s$.
Hence, $\langle X_s, Y_s \rangle$ satisfies all the four conditions with respect to instance $((G', k', d'), X', \langle X'_L, X_R \rangle)$.
This implies that $((G', k', d'), X', \langle X'_L, X_R \rangle)$ is a \yes-instance.

$(\Leftarrow)$ Suppose that $((G', k', d'), X', \langle X'_L, X_R\rangle)$ is a \yes-instance.
Then, there exists a solution pair $\langle X'_s, Y'_s \rangle$ that satisfies the four conditions mentioned in the definition of the problem.
Any edge in $G - F_1$ which is not present in $G'$ is incident on some vertex in $V(F_2)$.
By the first condition, the set $(X' \setminus X'_s) \cup Y'_s$ is a  vertex cover of $G'$.
As $(X' \setminus X_s) \cup Y_s$ is obtained from $(X \setminus X_s) \cup Y_s$ by contracting the edges in $F_2$ whose both endpoints are in $X \setminus X_s$, $(X \setminus X'_s) \cup Y'_s$ is a vertex cover of $G - F_1$.
This  implies that if $\rank((X' \setminus X'_s) \cup Y'_s) \ge k'$, then $\rank((X \setminus X'_s) \cup Y'_s) \ge k' + |F_2| = k$.
For every edge in $F_1$, one of its {endpoints} is incident on $X_L \subseteq X$.
Hence, $(X \setminus X'_s) \cup Y'_s$ is a vertex cover of $G$ and its rank is at least $k$.
As $k - d = k' - d'$, we have $|Y'_s| - |X'_s| \le k - d$.
It is easy to verify that $X_L \cap X'_s$ and $X_R \subseteq X'_s$.
Hence, $\langle X'_s, Y'_s \rangle$ satisfies all the four conditions with respect to the instance $((G, k, d), X, \langle X_L, X_R \rangle)$.
This implies that $((G, k, d), X, \langle X_L, X_R \rangle)$ is a \yes-instance.
\end{proof}

\subsubsection{Reducing to \textsc{Constrained MaxCut}}
\label{sub-sec:const-vc-to-const-max-cut}

We find the following {reformulation} of \textsc{Annotated Contraction(\vc)} convenient to present an algorithm to solve it.

\defproblem{\textsc{Constrained MaxCut}}{An instance $(G, k, d)$ of \textsc{Contraction(\vc)},  a minimum vertex cover $X$ of $G$, and a tuple $\langle X_L, X_R \rangle$ such that $X_L, X_R$ are disjoint subsets of  $X$.}{Does there exist a partition $\langle V_L, V_R \rangle$ of $V(G)$ such that
$(i)$ $E(V_L \cap Y, V_R \cap X) = \emptyset$,
$(ii)$ $\rank(E(V_L \cap X, V_R \cap Y)) \ge k$,
$(iii)$ $ |V_R \cap Y| - |V_R \cap X|  \le k - d$, and
$(iv)$ $X_L \subseteq V_L$ and $X_R \subseteq V_R$?}

Note that in \textsc{Annotated Contraction(\vc)} we are seeking {for} a pair of subsets, whereas in \textsc{Constrained MaxCut} we are looking for a partition of $V(G)$.
Such a formulation allows us to handle vertices that we have decided to keep out of {a} solution pair.
Note that the input instances for both of these problems are {the} same.
Hence, due to \autoref{lemma:rr-anno-contr-edges-X}, it is safe to assume that $X$ is a minimum vertex cover and an independent set in $G$.
{In the next lemma we show that both problems are in fact equivalent.} %\igm{I wonder wheth9er we could not just define the latter problem, and avoid the former. I guess we don't have time for this now, but we may keep it in mind}

\begin{lemma}
\label{lemma:red-const-contr-vc-maxcut}
An instance $((G, k, d), X, \langle X_L, X_R \rangle)$ is a {\sc Yes}-instance of {\sc Annotated Contraction(\vc)} if and only if it is a {\sc Yes}-instance of {\sc Constrained MaxCut}.
\end{lemma}
\begin{proof}
$(\Rightarrow)$ Let $\langle X_s,  Y_s \rangle$ be a solution pair of $((G, k, d), X, \langle X_L, X_R \rangle)$ for the \textsc{Annotated Contraction(\vc)} problem.
Define $V_R = X_s \cup Y_s$ and $V_L = V(G) \setminus V_R$.
It is easy to verify that $\langle V_L, V_R \rangle$ satisfies all the four conditions mentioned in the problem statement of \textsc{Constrained MaxCut}.
This implies that $((G, k, d), X, \langle X_L, X_R \rangle)$  is a \yes-instance of \textsc{Constrained MaxCut}.

$(\Leftarrow)$ Let $\langle V_L, V_R \rangle$ be a desired partition of $V(G)$.
Define $X_s = V_R \cap X$ and $Y_s = V_R \cap Y$.
Once again, it is easy to verify that the solution pair $\langle X_s, Y_s \rangle$ satisfies all the four conditions mentioned in problem statement of \textsc{Annotated Contraction(\vc)}.
This implies that $((G, k, d), X, \langle X_L, X_R \rangle)$ is a \yes-instance of \textsc{Annotated Contraction(\vc)}.
\end{proof}

Consider an instance $((G, k, d), X, \langle X_L, X_R \rangle)$ of \textsc{Constrained MaxCut}.
We consider the following two cases: $(1)$ $k = d$,  and $(2)$ $d < k < 2d$.
(Recall that we are in the case where $k < 2d$.)
The first case, as we will see, allows us to impose additional restrictions on the vertices that are in $V_R$.
It also helps us to set up some conditions such that, if they are satisfied while running the algorithm, then it can terminate and safely conclude that {the} input is a \yes-instance.
In the second case,  even for $k = d + 1$, we do not have these privileges.
We deal with each of the two cases separately.
\autoref{lemma:k-neq-d-to-k-eq-d} states that if an input instance is of the second type,  then we can construct a collection of $2^{\calO(d)} \cdot n^{k - d}$ many instances of the first type such that input instance is a \yes-instance if and only if at least one of these newly created instances is a \yes-instance.
We remark that this is the only place, in the whole algorithm, where {a} $n^{k - d}$-factor {appears} in the running time.
Recall that \autoref{thm:w1-hard} {implies} that this factor is unavoidable.
In the next subsection, we present an algorithm to solve the instances that are of the first type.

\begin{lemma}
\label{lemma:k-neq-d-to-k-eq-d}
Suppose that there is an algorithm that, given an instance $((G, k, d), X, \langle X_L, X_R \rangle)$ of {\sc Constrained MaxCut} with a guarantee that $k = d$, runs in time $f(n, k, d)$ and correctly determines whether it is a {\sc Yes}-instance.
Then, there is an algorithm that solves {\sc Constrained MaxCut} in time $f(n, k, d) \cdot 2^{\calO(d)} \cdot n^{k - d+1}$.
%\igm{This should be $n^{k - d}$, as it follows from that proof, and as we claim in the overview of the introduction}
\end{lemma}
\begin{proof}
Let $\calA$ be an algorithm that, given an instance $((G, k, d), X, \langle X_L, X_R \rangle)$ of \textsc{Constrained MaxCut} with the guarantee that $k = d$,  runs in time $f(n, k, d)$, and correctly determines whether it is a \yes-instance.
We describe an algorithm that solves \emph{any} instance of \textsc{Constrained MaxCut} using $\calA$ as a subroutine.

The algorithm takes as input an instance $((G, k, d), X, \langle X_L, X_R \rangle)$ of \textsc{Constrained MaxCut} and either returns \yes\ or \no.
For every input instance, it constructs a collection $\calI = \{(G^{i}, k^{i}, d^{i}), X^{i}, \langle X^{i}_L, X^{i}_R \rangle)\}$ of $2^{\calO(d)} \cdot n^{k - d}$ many instances,  as described below, such that $k^{i} = d^{i}$ for every $i \in [|\calI|]$.
It uses Algorithm $\calA$ to check if at least one of these instances is {a} \yes-instance.
If it is the case then it returns \yes, otherwise it returns \no.

The algorithm constructs the new instances as follows.
First it guesses an integer $q$ in $\{0, 1, \dots, k - d\}$ such that if $((G, k, d), X, \langle X_L, X_R \rangle)$ is a \yes-instance then $q$ is the smallest integer for which $((G, d + q, d), X, \langle X_L, X_R \rangle)$ is a \yes-instance.
Note that the solution size, i.e., the number of edges allowed to be contracted in the second instance, is $d + q$.
Let $\calY_{q} = \{Y^i \subseteq Y\ |\ |Y^i| = q\}$.
For every set $Y^i$ in the collection, the algorithm constructs at most $2^{\calO(k)}$ many new instances.
If $\rank(E(Y^i, N(Y^i))) \ge k$, then the algorithm constructs a trivial \yes-instance.
It constructs a graph $G^i$ as follows:
For every vertex $y \in Y^i$, it adds a vertex $x$ and makes it adjacent to $y$.
Alternately, every vertex $y$ in $Y^{i}$ is adjacent to a pendant vertex in $G^i$.
Let $X^i_{p}$ be the collection of all the pendant vertices added while constructing $G^i$ from $G$.
%Define $X^{i} = X \cup X^i_p$.
%If $X^i$ is \emph{not} a minimum vertex cover of $G^i$, then the algorithm constructs a trivial \no-instance.
%Otherwise,
For every partition $\langle X^i_{\ell}, X^i_r \rangle$ of $N(Y^i)$, the algorithm constructs a new instance $((G^{i}, k^{i}, d^{i}), X^{i}, \langle X^{i}_L, X^{i}_R \rangle)$ where $G^i$ is as described above, $k^{i} := d + q$, $d^{i} := d + q$, $X^{i} = X \cup X^i_p$, $X^i_L := X_L \cup X^i_{\ell}$, and $X^i_R := X_R \cup X^i_r \cup X^i_p$.
This concludes the description of the algorithm.

We now argue the correctness of the algorithm.
Suppose that $((G, k, d), X, \langle X_L, X_R \rangle)$ is a \yes-instance.
Hence, there exists an integer $q$ in $\{0, 1, \dots, k - d\}$ such that $((G, d + q, d), X, \langle X_L, X_R \rangle)$ is a \yes-instance.
We assume, without loss of generality, that $q$ is the smallest such an integer.
Consider the case when there exists a subset $Y'$ of $Y$ such that $|Y'| = q$ and $\rank(E(Y', N(Y'))) \ge k$.
In this case, the algorithm constructs a trivial \yes-instance, and thus correctly concludes that the input is a \yes-instance.
Now consider the case when for every subset $Y'$ of $Y$ of size $q$, $\rank(E(Y', N(Y'))) < k$.
Suppose that $\langle V_L, V_R \rangle$ is a partition of $V(G)$ that satisfies all the conditions mentioned in the statement of the problem.
By the third condition, $|V_R \cap Y| - |V_R \cap X| = q$.
Fix a subset $Y^{i}$ of $V_R \cap Y$ of size~$q$.
The third condition ensures that such a set exists.
Consider a partition $\langle X^i_{\ell}, X^i_{r} \rangle$ of $N(Y^i)$ where $X^i_{\ell} = N(Y^i) \cap V_L$ and $X^i_{r} = N(Y^i) \cap V_R$.
As the algorithm constructs an instance for every subset $Y'$ of $Y$ of size $q$, and for every partition of $N(Y')$, it constructs an instance, say $((G^{i}, k^{i}, d^{i}), X^{i}, \langle X^{i}_L, X^{i}_R \rangle)$, corresponding to $Y^i$ and a partition $\langle X^i_{\ell}, X^i_{r} \rangle$ of $N(Y^i)$.
We argue that this is a \yes-instance.

Let $X^i_p$ be the set of pendant vertices added to $G$ to construct $G^i$.
Note that $|X^i_p| = q$.
Hence, $V(G^i) = V(G) \cup X^i_p$.
Consider a partition $\langle V^i_L, V^i_R \rangle$ of $V(G^i)$ where $V^i_L = V_L$ and $V^i_R = V_R \cup X^{i}_p$.
It is easy to verify that this partition satisfies all the four conditions mentioned in the problem statement.
This implies that $((G^{i}, k^{i}, d^{i}), X^{i}, \langle X^{i}_L, X^{i}_R \rangle)$ is a \yes-instance.
By the correctness of Algorithm~$\calA$, it correctly concludes that it is a \yes-instance.
Hence, in this case the algorithm returns \yes\ on at least one of the newly created instances.
Hence, if the input  is a \yes-instance, then the algorithm correctly concludes that it is a \yes-instance.

We now argue that if the algorithm returns \yes, then the input instance is indeed a \yes-instance.
Consider an input instance $((G, k, d), X, \langle X_L, X_R \rangle)$.
By the description of the algorithm, it returns \yes\ if and only if one of the newly created instances is a \yes-instance.
Suppose that one of these instances is a trivial \yes-instance.
The algorithm constructs such an instance only if it finds a subset $Y'$ of $Y$ which is of size at most $k - d$, and $\rank(N(Y'), Y') \ge k$.
In this case, $\langle V_L = V(G) \setminus Y', V_R = Y' \rangle$ satisfies all the conditions mentioned in the problem statement.
Hence, the input instance is a \yes-instance.
Otherwise, suppose that the algorithm concludes that a non-trivial instance $((G^{i}, k^{i}, d^{i}), X^{i}, \langle X^{i}_L, X^{i}_R \rangle)$ is a \yes-instance using Algorithm~$\calA$.
Suppose that $\langle V^i_L, V^i_R \rangle$ is a partition of $V(G^i)$ that satisfies all the four conditions in the problem statement.
Let $X^i_p$ be the collection of pendant vertices added while constructing $G^{i}$ from $G$.
Recall that $X^i_p \subseteq X^i_r$ and $|X^i_p| \le k - d$.
It is easy to verify that $\langle V_L = V^i_L, V_R = V^i_R \setminus X^i_p \rangle$ is the desired partition of $V(G)$ that satisfies all the four conditions mentioned in the problem statement.
This implies that the input is a \yes-instance.
This concludes the proof of correctness of the algorithm.

It remains to argue about the running time of the algorithm.
For the input instance $((G, k, d), X, \langle X_L, X_R \rangle)$, there are at most $(k - d) + 1$ choices for $q$.
For each $q$, the size of $\calY_q$, the collection of subsets of $Y$ of size exactly $q$, is at most $n^{q}$.
If for a set $Y'$ in $\calY_q$, $\rank(Y', N(Y')) \ge k$, then the algorithm creates only one instance.
Otherwise, $\rank(Y', N(Y')) < k$.
By \autoref{obs:rank-bound-vertices}, the number of vertices in $N(Y')$ is bounded by $k$.
In this case, the algorithm constructs $2^{\calO(k)}$ many instances.
The overall running time of the algorithm follows from the running time of Algorithm $\calA$ and the fact that $k \le 2d$.
\end{proof}

\subsubsection{Simplifying an instance of \textsc{Constrained MaxCut} when $k = d$}
\label{subsubsec-simplify-k-eq-d}

As mentioned before, in this subsection we present an algorithm to solve an instance $((G, k, d), X, \langle X_L, X_R \rangle)$ of \textsc{Constrained MaxCut} with a guarantee that $k = d$.
We first present a reduction rule to simplify these instances under the presence of a matching saturating $X$, and prove its correctness using the fact that $k = d$.

\begin{reduction rule}
\label{rr:only-matching}
Consider an instance $((G, k, d), X, \langle X_L, X_R \rangle)$ of {\sc Constrained MaxCut} such that $k = d$ and $X$ is an independent set in $G$.
Let $M$ be a matching in $G$ saturating $X$.
\begin{itemize}
\item If there exists $x \in X \setminus X_L$ such that $N(x) \setminus V(M) \neq \emptyset$, then add $x$ to $X_L$.
\item If there exists $x \in X_L$ such that $N(x) \setminus V(M) \neq \emptyset$, then delete all vertices in $N(x) \setminus V(M)$.
\end{itemize}
Return instance $((G', k, d), X, \langle X'_L, X_R \rangle)$ where $G' = G - (N(x) \setminus V(M))$ and $X'_L = X_L \cup \{x\}$.
\end{reduction rule}

\begin{lemma}
\label{lemma:rr-only-matching}
\autoref{rr:only-matching} is safe.
\end{lemma}
\begin{proof}
Suppose that $((G, k, d), X, \langle X_L, X_R \rangle)$ is a \yes-instance, and let $\langle V_L, V_R \rangle$ be the desired partition of $V(G)$ that satisfies all the conditions in the problem statement.
By the first condition, $N(X \cap V_R) \subseteq Y \cap V_R$.
Let $M'$ be the subset of edges of $M$ that saturates all vertices in $X \cap V_R$.
As $V(M') \cap Y \subseteq N(X \cap V_R)$, we have $(V(M') \cap Y) \subseteq (Y \cap V_R)$.
As $M$ is a matching, $|X \cap V_R| = |M'| = |V(M') \cap Y|$.
However, as $k = d$, the third condition implies $|Y \cap V_R| - |X \cap V_R| = 0$.
Hence, $(Y \cap V_R) \setminus (V(M') \cap Y)$ is an empty set.

Consider the first case.
If $x \in V_R$ then, by the first condition, $N(x) \setminus M$ is in $Y \cap V_R$.
However, this contradicts the fact that $(Y \cap V_R) \setminus (V(M') \cap Y)$ is an empty set.
This implies that if $((G, k, d), X, \langle X_L, X_R \rangle)$ is a \yes-instance, then $((G', k, d), X, \langle X'_L, X_R \rangle)$ is also a \yes-instance.
The reverse direction is vacuously true.

Consider the second case.
If $N(x) \setminus V(M)$ is in $V_R$, this contradicts the fact that $(Y \cap V_R) \setminus (V(M') \cap Y)$ is an empty set.
Hence, for any partition $\langle V_L, V_R \rangle$, $N(x) \setminus V(M)$ is in $V_L$.
It is easy to see that $\langle V_L, V_R \rangle$ is a solution of $((G, k, d), X, \langle X_L, X_R \rangle)$ if and only if $\langle V_L \setminus (N(x) \setminus M), V_R \rangle$ is a solution of $((G', k, d), X, \langle X'_L, X_R \rangle)$.
\end{proof}

\begin{figure}[t]
  \begin{center}
    \includegraphics[scale=0.6]{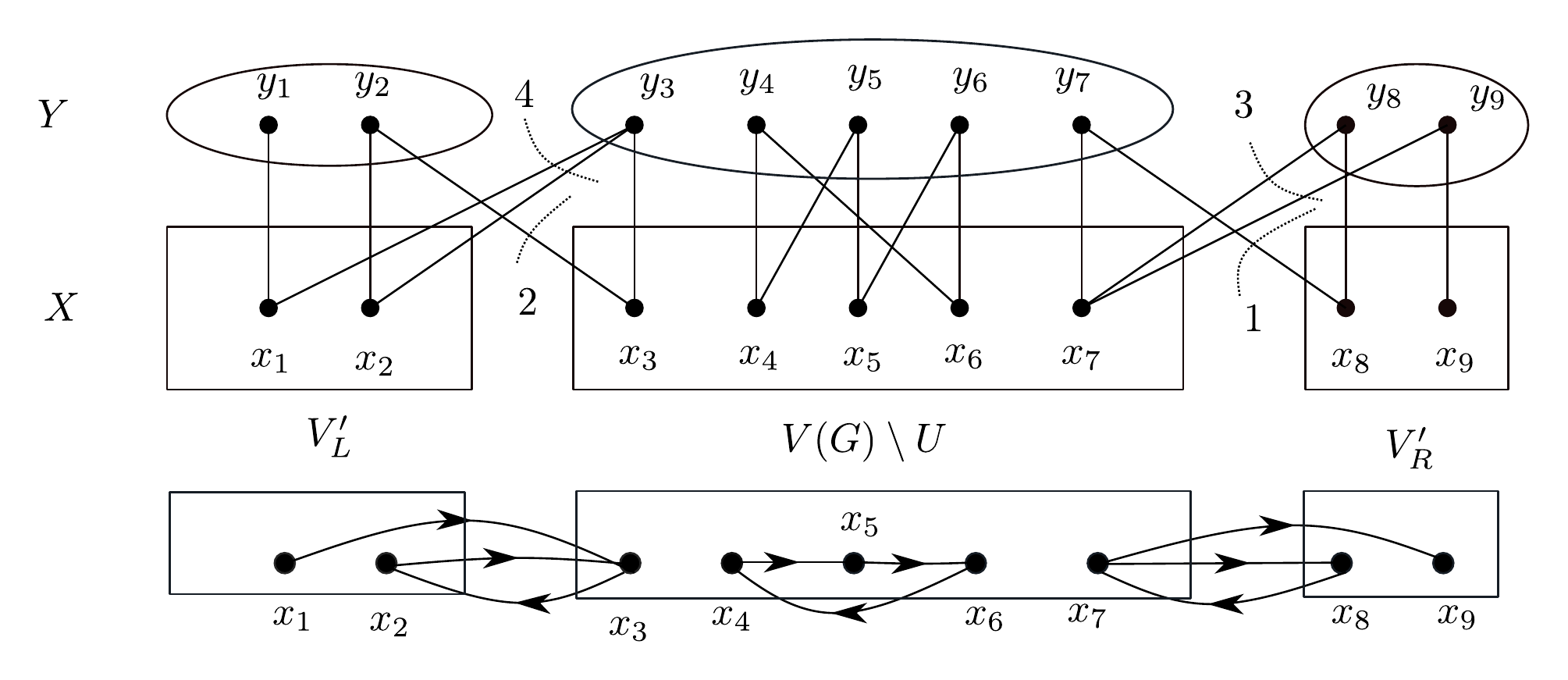}
    \end{center}
   \caption{An example to illustration a reduction from \textsc{Constrained MaxCut} to \textsc{Contrained Directed MaxCut}. We do not show all the edges in $G$ for the sake of clarity. \label{fig:vc-to-const-dir-max-cut}}
\end{figure}

We now present an informal description of the algorithm {to solve \textsc{Constrained MaxCut} with a guarantee that $k = d$}.
Consider an instance $((G, k, d), X, \langle X_L, X_R \rangle)$ of \textsc{Constrained MaxCut} on which \autoref{rr:only-matching} is not applicable.
See \autoref{fig:vc-to-const-dir-max-cut} for an illustration.
Let $M = \{x_iy_i\ |\ i \in [9]\}$ be a matching saturating the vertices in $X$.
Note that, in this case, $|X| = |Y| = |M|$.
Consider a subset $U$ of $V(G)$ which can be `well-partitioned' into $\langle V'_L, V'_R \rangle$.
Informally, this means that $\langle V'_L, V'_R \rangle$ can be extended to obtain a partition $\langle V_L, V_R \rangle$ of $V(G)$ such that it is a solution of $((G, k, d), X, \langle X_L, X_R \rangle)$.
We can think of the vertices in $U$ as `processed vertices'.
For example, consider $U = \{x_1, y_1, x_2, y_2, x_8, y_8, x_9, y_9\}$ in \autoref{fig:vc-to-const-dir-max-cut},
and let $\langle V'_L, V'_R \rangle$ be a `well-partition' of $U$ where $V'_L = \{x_1, y_1, x_2, y_2\}$ and $V'_R = \{x_8, y_8, x_9, y_9\}$.
Our objective is to extend $V'_L, V'_R$ to obtain $V_L, V'_R$ by processing more vertices, i.e., by adding them to either $V'_L$ or $V'_R$.

As $G$ is connected and $X, Y$ are independent sets in $G$,  at least one of the following four sets is non-empty:
$(1)$ $E(Y \setminus U,  V'_R \cap X)$,
$(2)$ $E(X \setminus U,  V'_L \cap Y)$,
$(3)$ $E(X \setminus U,  V'_R \cap Y)$, and
$(4)$ $E(Y \setminus U,  V'_L \cap X)$.
As we are aiming for a partition $\langle V_L, V_R \rangle$ for which $E(V_L \cap Y, V_R \cap X) = \emptyset$, in the first case it is safe to move the endpoints of the edges in $E(Y \setminus U,  V_R \cap X)$ that are in $Y \setminus U$ to $V_R$.
For example, it is safe to move $y_7$ to $V'_R$.
Similarly, in the second case, it is safe to move the endpoints of edges in $E(X \setminus U,  V_L \cap Y )$ that are in $X \setminus U$ to $V_L$.
For example, it is safe to move $x_3$ to $V'_L$.
As $\langle V_L, V_R \rangle$ also needs to satisfy $|Y \cap V_R| = |X \cap V_R|$, such a move also forces other vertices that are adjacent to these vertices via edges in $M$ to move.
For example, $x_7$ and $y_3$ are forced to move to $V'_R$ and $V'_L$, respectively.

In the third case, if $\rank(E(X \setminus U,  V_R \cap Y)) \ge k$ then $\langle V'_L \cup (V(G) \setminus U), V'_R \rangle$ is the desired partition.
Otherwise,  $\rank(E(X \setminus U,  V_R \cap Y)) < k$.
\iffalse
By \autoref{obs:rank-bound-vertices},
\fi
It is easy to see that in this
case, the number of vertices in $X \setminus U$ that are incident on edges in $E(X \setminus U,  V_R \cap Y)$ is at most $k$.
We can guess how the desired partition $\langle V_L, V_R \rangle$ intersects with these endpoints and extend the set of processed vertices in $2^{\calO(k)}$ many ways.
Similarly, in the fourth case, if $\rank(E(Y \setminus U,  V_L \cap X)) \ge k$, then $\langle V'_L, V'_R \cup (V(G) \setminus U)\rangle$ is a partition that satisfies all the desired conditions.
Otherwise, we can extend to a set of processed vertices in $2^{\calO(k)}$ many ways.

To implement the idea mentioned in the above paragraph, we need some bound on the total number of sets of `processed vertices' we need to consider.
In order to do that, we exploit the properties of the desired partition.
Consider the set $\{x_4, y_4, x_5, y_5, x_6, y_6\}$  in \autoref{fig:vc-to-const-dir-max-cut}.
Because of the arguments used in the first and the second cases, either this  set is entirely contained in $V_L$ or in $V_R$.
To find such cycles, we introduce a directed version of the problem called \textsc{Constrained Digraph MaxCut}.
The input of the problem contains a digraph and the objective is to find a partition of the vertex set such that all the arcs across this partition are in the same direction, and the rank of these arcs is at least $k$.
For our case, consider the digraph $D$ obtained from $G$ by directing every edge from $X$ to $Y$ and then `merging' all edges in the matching $M$.
Recall that $M$ is a matching saturating $X$.
Here, we do not delete parallel or anti-parallel edges while merging an arc in a directed graph.
See \autoref{fig:vc-to-const-dir-max-cut} for an example.
In $D$, sets like these correspond to a directed cycle.
And as mentioned before, vertices in these directed cycles move together.
Hence, we can obtain a directed acyclic graph by merging these cycles into a vertex.
The topological ordering of this resulting graph gives a natural order to process the vertices in $G$.
We formalize these ideas in the next subsection.

\subsubsection{Reducing to \textsc{Constrained Digraph MaxCut}}
\label{sub-sec:cont-max-cut-dir-max-cut}

In this section, we consider directed graphs that can have parallel arcs.
For a digraph $D$, we define its \emph{underlying undirected graph} $G$ as the graph obtained from $D$ by forgetting the directions of the arcs.
Formally, $V(G) = V(D)$ and $E(G) = \{uv\ |\ (u, v) \in A(D) \}$.
We define the \emph{rank} of a digraph, and the rank of a subset of its vertices or arcs using its underlying undirected graph.
Formally, $\rank(D) = \rank(G)$, for a subset $S \subseteq V(D)$, $\rank(S) = \rank(G[S])$, and for a subset $B \subseteq A(D)$, $\rank(B) = \rank(G[V(B)])$.

\defproblem{\textsc{Constrained Digraph MaxCut}}{A digraph $D$, a tuple $\langle X_L, X_R \rangle$ of disjoint subsets of $X$, and an integer $k$.}{Does there exist a partition $(V_L, V_R)$ of $V(G)$ such that
$(i)$ $A(V_R, V_L) = \emptyset$,
$(ii)$ $\rank(A(V_L, V_R)) \ge k$, and
$(iii)$ $X_L \subseteq V_L$ and $X_R \subseteq V_R$?
}

We say that a partition $\langle V_L, V_R\rangle$ is a \emph{solution} of $(D, \langle X_L, X_R \rangle, k)$ if it satisfies all the three conditions in the statement of the problem.
We present a reduction that, given an instance $((G, k, d), X, \langle X_L, X_R \rangle)$ of \textsc{Constrained MaxCut}, returns an instance $(D, \langle X_L, X_R \rangle, k)$ of \textsc{Constrained Directed MaxCut}.

\medskip

\noindent\textbf{The reduction:}
\sloppy
The reduction takes as input an instance $((G, k, d), X, \langle X_L, X_R \rangle)$ of \textsc{Constrained MaxCut} on which \autoref{rr:only-matching} is not applicable.
It starts with a copy of the graph $G$ and constructs a digraph $D$.
The reduction finds (in polynomial time) a matching $M$ in $G$ that saturates all vertices in $X$.
For every $xy \in E(G)$, where $x \in X$ and $y \in Y$, it deletes edge $xy$ and adds arc $(x, y)$ (i.e., it directs edges from $X$ to $Y$).
For every arc $(x, y)$ in $M$, it does as follows:
For every in-neighbour $x_1$ of $y$, it adds arc $(x_1, x)$.
It then deletes vertex $y$.
This completes the construction of digraph $D$.
The reduction returns $(D, \langle X_L, X_R \rangle, k)$ as the instance of \textsc{Constrained Digraph MaxCut}.
This completes the description of the reduction.

\begin{lemma}
\label{lemma:redu-maxcut-digraph-maxcut}
$((G, k, d), X, \langle X_L, X_R \rangle)$ is a \yes-instance of {\sc Constrained MaxCut} if and only if $(D, \langle X_L, X_R \rangle, k)$ as a {\sc Yes}-instance of {\sc Constrained Digraph MaxCut}.
\end{lemma}
\begin{proof}
Note that the edges in $E(V_L \cap Y, V_R \cap X)$ correspond to arcs in $A(V_R \cap X, V_L \cap X)$ and the edges in $E(V_L \cap X, V_R \cap Y)$ correspond to $A(V_L \cap X, V_R \cap X)$.
Moreover, by the construction of $D$, two edges in $E(V_L \cap X, V_R \cap Y)$ share an endpoint if and only if the corresponding two arcs in $A(V_L \cap X, V_R \cap X)$ share an endpoint.
This implies that $\rank(E(V_L \cap X, V_R \cap Y)) = \rank(A(V_L \cap X, V_R \cap X))$.

Consider a solution $\langle V_L, V_R \rangle$ of $((G, k, d), X, \langle X_L, X_R \rangle)$.
It is easy to see that $\langle V_L \cap X, V_R \cap X \rangle$ is a solution of $(D, \langle X_L, X_R \rangle, k)$.

Similarly, consider a solution $\langle V'_L, V'_R \rangle$ of $(D, \langle X_L, X_R \rangle, k)$.
As $V(D) = X$, $\langle V'_L, V'_R \rangle$ is a partition of $X$.
Let $V_L$ be the collection of vertices in $V'_L$ as well as vertices in $Y$ that are adjacent to vertices in $V'_L$ via edges in $M$.
It is easy to verify that $\langle V_L, V_R = V(G) \setminus V_L \rangle$ is a solution of $((G, k, d), X, \langle X_L, X_R \rangle)$.
\end{proof}

Consider a digraph $D^{\circ}$ obtained from $D$ by \emph{merging} a directed cycle $C$ into a single vertex in $D$.
Formally, this operation adds a vertex $x_C$ to $V(D)$, and for every arc $(x, x_1)$ in $A(V(C), V(D) \setminus V(C))$, it adds arc $(x_C, x_1)$ and for every arc $(x_1, x)$ in $A(V(D) \setminus V(C))$, it adds arc $(x_1, x_C)$.
Note that, unlike with the edge contraction operation, this operation does not delete parallel or anti-parallel arcs.
Consider a map $\psi: V(D) \mapsto V(D^{\circ})$ where $\psi(x) = x$ or it is the vertex added to $V(D^{\circ})$ while merging a directed cycle containing $x$.
It is easy to verify that $\psi$ defines a surjective function.
Consider a directed acyclic graph $D'$ obtained from $D$ by repeatedly merging directed cycles.
Let $D = D_1, D_2, \dots, D_q = D'$ be the sequence of the digraphs such that $D_{i + 1}$ is obtained by merging a cycle $C_i$ in $D_i$, and let $\psi_i : V(D_i) \mapsto V(D_{i + 1})$ be the function as defined above.
We define $\psi: V(D) \mapsto V(D')$ inductively, i.e., $\psi(x) = \psi_{q}(\psi_{q-1}(\cdots (\psi_1(x))))$.
Once again, it is easy to verify that $\psi$ defines a surjective function.
For any $x' \in V(D')$, we define $\psi^{-1}(x') := \{x \in V(D)\ |\ \psi(x) = x'\}$, and for any subset $U \subseteq V(D')$, $\psi^{-1}(U) := \bigcup_{x' \in U} \psi^{-1}(x')$.
A \emph{topological ordering} of a directed acyclic graph $D'$ is a linear ordering $\sigma: V(D') \mapsto [|V(D')|]$ such that for every arc $(x, x_1)$, $\sigma(x) < \sigma(x_1)$.

We are now in position to present an algorithm for \textsc{Constrained Directed MaxCut}.

\begin{lemma}
\label{lemma:algo-const-dir-maxcut}
There is an algorithm that, given an instance $(G, \langle X_L, X_R\rangle, k)$ of {\sc Constrained Directed MaxCut}, runs in time $2^{\calO(k)} \cdot n^{\calO(1)}$ and correctly determines whether it is a {\sc Yes}-instance.
\end{lemma}
\begin{proof}
The algorithm takes as input an instance $(D, \langle X_L, X_R \rangle, k)$ of \textsc{Constrained Directed MaxCut}, and returns either \yes\ or \no.
It starts by constructing a directed acyclic graph $D'$ from $D$ by merging directed cycles in $D$ as described above.
Suppose that $D'$ is the resulting directed acyclic graph obtained after this procedure.
The algorithm applies the following two reduction rules exhaustively.
\begin{itemize}
\item If there is an arc $(x_1, x)$ such that $x \in X_L$, then delete arc $(x_1, x)$, and add $x_1$ to $X_L$.
\item If there is an arc $(x, x_1)$ such that $x \in X_R$, then delete arc $(x, x_1)$, and add $x_1$ to $X_R$.
\end{itemize}
The correctness of these reduction rules follows from the fact that, for any solution $\langle V_L, V_R \rangle$ of $(G, \langle X_L, X_R \rangle, k)$, $X_L \subseteq V_L$, $X_R \subseteq V_R$, and $A(V_R, V_L) = \emptyset$.

For notational convenience, we denote the resulting instance obtained after exhaustive application of these reduction rules by $(D', \langle X_L, X_R \rangle, k)$.
Note that every vertex in $X_L$ has in-degree zero and every vertex in $X_R$ has out-degree zero.
Let $n'$ be the number of vertices in $V(D')$ and $\sigma = \{x_1, x_2, \dots, x_{n'}\}$ be a topological ordering of $V(D')$.

We assume, without loss of generality, that the first $|X_L|$ vertices are in $X_L$, and that the last $|X_R|$ vertices are in $X_R$.
For every $i \in [n']$, define $U^i := \{u_i, u_{i + 1}, \dots, u_{n'}\}$, and $W^{i} \subseteq \psi^{-1}(U^i)$ as the collection of endpoints of arcs in $A(\psi^{-1}(V(D) \setminus U^i), \psi^{-1}(U^i) )$ that are in $\psi^{-1}(U^i)$.

If there is an integer $i \in [n']$, such that $X_L \subseteq  \psi^{-1}(V(D) \setminus U^i)$, $X_R \subseteq \psi^{-1}(U^i)$, and the rank of the arcs across the partition $\langle \psi^{-1}(V(D) \setminus U^i), \psi^{-1}(U^i) \rangle$ is at least $k$, then the algorithm returns \yes.
Otherwise, the algorithm constructs a dynamic programming table
%whose entries are indexed by $(i, \langle W_{L}^{i}, W_R^{i} \rangle, k^{\circ})$
such that $\calT[i, \langle W_{L}^{i}, W_R^{i} \rangle, k^{\circ}]$ is a `valid partition' $\langle V_{L}^{i}, V_R^{i} \rangle$ of $U^i$ which is `consistent' with $\langle W_{L}^{i}, W_R^{i} \rangle$, and has rank at least $k^{\circ}$, i.e., $\rank(A(V_{L}^{i}, V_R^{i})) \ge k^{\circ}$, if such a partition exists, otherwise it is $\langle \emptyset, \emptyset \rangle$.
Here, $i \in [n']$, $\langle W_{L}^{i}, W_R^{i} \rangle$ is a partition of $W^{i}$, and $k^{\circ} \in \{0\} \cup [k]$. The details follow.

\begin{definition}
We say that  a partition $\langle V_{L}^{i}, V_R^{i} \rangle$ of $U^i$ is a \emph{valid partition} if $(i)$ $A(V_R^{i}, V_{L}^{i}) = \emptyset$, and $(ii)$ $(X_L \cap U^i) \subseteq V_{L}^{i}$ and $(X_R \cap U^i) \subseteq V_{R}^{i}$.
We say that a valid partition $\langle V_{L}^{i}, V_R^{i} \rangle$ of $U^i$ is \emph{consistent} with $\langle W_{L}^{i}, W_R^{i} \rangle$ if $W_{L}^{i} \subseteq V_{L}^{i}$ and $W_R^{i} \subseteq V_R^{i}$.
\end{definition}

The algorithm initializes $\calT[i, \langle W_{L}^{i}, W_R^{i} \rangle, k^{\circ}]$ to $\langle \emptyset, \emptyset \rangle$ for every $i \in [n']$, for every partition $\langle W_{L}^{i}, W_R^{i} \rangle$ of $W^{i}$, and for every $k^{\circ} \in \{0\} \cup [k]$.
It sets the following values:

\begin{itemize}
\item $\calT[n', \langle W_{L}^{n} = \{u_{n'}\}, W_R^{n} = \emptyset \rangle, k^{\circ} = 0]$ to $\langle V_L^{n} = \{u_{n'}\}, V_R^n = \emptyset\rangle$, and
\item $\calT[n', \langle W_{L}^{n} = \emptyset, W_R^{n} = \{u_{n'}\} \rangle, k^{\circ} = 0]$ to $\langle V_L^{n}  = \emptyset, V_R^n= \{u_{n'}\}\rangle$.
\end{itemize}

To compute $\calT[i, \langle W_{L}^{i}, W_R^{i} \rangle, k^{\circ}]$, the algorithm considers the following three cases: $(1)$ $u_{i} \in W_L^i$, $(2)$ $u_{i} \in W_R^i$, and $(3)$ $u_{i} \not\in W_L^i \cup W_R^i$.
\medskip

\begin{enumerate}
\item In the first case, if there exists a table entry $(i + 1, \langle W_L^{i + 1}, W_R^{i + 1} \rangle, k')$ for some $k' \in \{0\} \cup [k]$ such that
\begin{itemize}
\item $W_L^{i} \setminus \{u_i\} \subseteq W_L^{i + 1}$ and $W_R^{i} \subseteq W_R^{i + 1}$, and
\item $\rank(A(V^{i+1}_L \cup \{u_i\}, V^{i+1}_R)) \ge k^{\circ}$,
\end{itemize}
where $\langle V^{i+1}_L, V^{i+1}_R \rangle = \calT[i + 1, \langle W_L^{i + 1}, W_R^{i + 1} \rangle, k']$, then the algorithm sets $\calT[i, \langle W_{L}^{i}, W_R^{i} \rangle, k^{\circ}] = \langle V^{i+1}_L \cup \{u_i\}, V^{i+1}_R \rangle$.
\medskip
\item In the second case, if there exists a table entry $(i + 1, \langle W_L^{i + 1}, W_R^{i + 1} \rangle, k')$ for some $k' \in \{0\} \cup [k]$ such that
\begin{itemize}
\item $W_L^{i} \subseteq W_L^{i + 1}$ and $W_R^{i} \setminus \{u_i\} \subseteq W_R^{i + 1}$,
\item $N_{\sf out}(u_i) \cap W_L^{i + 1} = \emptyset$, and
\item $\rank(A(V^{i+1}_L, V^{i+1}_R \cup \{u_i\})) \ge k^{\circ}$,
\end{itemize}
where $\langle V^{i+1}_L, V^{i+1}_R \rangle = \calT[i + 1, \langle W_L^{i + 1}, W_R^{i + 1} \rangle, k']$, then the algorithm sets $\calT[i, \langle W_{L}^{i}, W_R^{i} \rangle, k^{\circ}] = \langle V^{i+1}_L, V^{i+1}_R \cup \{u_i\} \rangle$.
\medskip
\item In the third case, if there exists a table entry $(i + 1, \langle W_L^{i + 1}, W_R^{i + 1} \rangle, k')$ for some $k' \in \{0\} \cup [k]$ such that
\begin{itemize}
\item $W_L^{i} \subseteq W_L^{i + 1}$ and $W_R^{i} \subseteq W_R^{i + 1}$, and
\item $\rank(A(V^{i+1}_L \cup \{u_i\}, V^{i+1}_R)) \ge k^{\circ}$,
\end{itemize}
where $\langle V^{i+1}_L, V^{i+1}_R \rangle = \calT[i + 1, \langle W_L^{i + 1}, W_R^{i + 1} \rangle, k']$, then the algorithm sets $\calT[i, \langle W_{L}^{i}, W_R^{i} \rangle, k^{\circ}] = \langle V^{i+1}_L \cup \{u_i\}, V^{i+1}_R \rangle$.
\end{enumerate}

If $\calT[1, \langle W^1_L, W^1_R \rangle, k]$ is not $\langle \emptyset, \emptyset \rangle$ for some partition $\langle W^1_L, W^1_R \rangle$ of $W^1$, then the algorithm returns \yes, otherwise it returns \no.
This concludes the description of the algorithm.

We now argue that if the algorithm computes the above dynamic programming table, then all the entries are correct.
It is easy to verify that all the entries corresponding to indices where $i = n'$ are correct.
Suppose this is true for every integer in $[n']$ which is greater than $i$.
Consider an index $(i, \langle W^i_L, W^i_R \rangle, k^{\circ})$ and suppose that there exists a valid partition $\langle V^i_L, V^i_R \rangle$ of $U^i$ which is consistent with $\langle W^i_L, W^i_R \rangle$.
By the definition, $W_L^{i} \setminus \{u_{i}\} \subseteq W_L^{i + 1}$ and $W_R^{i} \setminus \{u_{i}\} \subseteq W_R^{i + 1}$.
We consider the following three exhaustive cases.

\begin{itemize}
\item If $u_i \in W^i_L$ (which implies $u_i \in V^i_L$) , then $\langle V^i_L \setminus \{u_i\}, V^i_R \rangle$ is a valid partition of $U^{i + 1}$ which is consistent with some partition $\langle W^{i + 1}_L, W^{i + 1}_R \rangle$ of $W^{i + 1}$, such that $W^{i}_L \setminus \{u_i\} \subseteq W^{i + 1}_L$ and $W^{i}_R \subseteq W^{i + 1}_R$.

\item If $u_i \in W^i_R$ (which implies $u_i \in V^i_R$), then $\langle V^i_L, V^i_R  \setminus \{u_i\} \rangle$ is a valid partition of $U^{i + 1}$ which is consistent with some partition $\langle W^{i + 1}_L, W^{i + 1}_R \rangle$ of $W^{i + 1}$, such that $W^{i}_L \subseteq W^{i + 1}_L$ and $W^{i}_R \setminus \{u_i\} \subseteq W^{i + 1}_R$.
Note that, as $N_{\sf out}(u_i) \subseteq W_L^{i + 1} \cup W_R^{i + 1}$, we have $N_{\sf out}(u_i) \cap W_L^{i + 1} = \emptyset$.

\item If $u_i \not\in W^i_L \cup W^i_R$, then there is no arc whose endpoint is $u_i$.
In this case, $\langle V^i_L \setminus \{u_i\}, V^i_R \rangle$ is a valid partition of $U^{i + 1}$ which is consistent with some partition $\langle W^{i + 1}_L, W^{i + 1}_R \rangle$ of $W^{i + 1}$, such that $W^{i}_L \subseteq W^{i + 1}_L$ and $W^{i}_R \setminus \{u_i\} \subseteq W^{i + 1}_R$.
\end{itemize}

As the algorithm correctly computes these values for every integer greater than $i$, it correctly computes the value of $\calT[i, \langle W^i_L, W^i_R \rangle, k^{\circ}]$.

We now argue about the correctness of the algorithm.
Suppose that the input instance $(D, \langle X_L, X_R \rangle, k)$ is a \yes-instance.
If there exists a partition $\langle V_L, V_R \rangle$ of $V(D)$ which is a solution, such that $V_R = \psi^{-1}(U_i)$ for some $i \in [n']$, then the rank of the arcs across the partition $\langle \psi^{-1}(V(D) \setminus U^i), \psi^{-1}(U^i) \rangle$ is at least $k$.
In this case, the algorithm correctly concludes that the input is a \yes-instance.
Otherwise, it computes the table as described above.
By the description of the algorithm, $U^1 = V(D')$, and the fact that all the entries in the table are correct, in this case the algorithm correctly concludes that the input is a \yes-instance.

Suppose that the algorithm returns \yes\ on the input instance $(D, \langle X_L, X_R \rangle, k)$.
Consider the case when the algorithm returns \yes\ without constructing the table.
Note that every arc across the partition $\langle V(D') \setminus U^i, U^i \rangle$ has its startpoint in $V(D') \setminus U^i$ and endpoint in $U^i$.
Alternately, $A( U^i, V(D') \setminus U^i) = \emptyset$.
By the construction of $D'$, $A(\psi^{-1}(U^i), V(D) \setminus \psi^{-1}(U^i)) = \emptyset$ for every $i \in [n']$.
As there exists $i \in [n']$ such that $X_L \subseteq  \psi^{-1}(V(D) \setminus U^i)$, $X_R \subseteq \psi^{-1}(U^i)$, and the rank of the arcs across the partition $\langle \psi^{-1}(V(D) \setminus U^i), \psi^{-1}(U^i) \rangle$ is at least $k$, $\langle V(D) \setminus U^i), \psi^{-1}(U^i) \rangle$ is a solution of $(G, \langle X_L, X_R \rangle, k)$.
In the other case, when the table is constructed, as every entry in the table is correct, the algorithm correctly concludes that the input is a \yes-instance.
This completes the proof of correctness of the algorithm.

To argue about the running time of the algorithm, notice that all the other steps, apart from computing the table, can be performed in polynomial time.
The algorithm computes the table if and only if for every integer $i \in [n']$, the rank of arcs across the partition $\langle \psi^{-1}(V(D) \setminus U^i), \psi^{-1}(U^i) \rangle$ is less than $k$.
By  arguments similar to those of \autoref{obs:rank-bound-vertices}, it is easy to see that $|W^{i}| < k$ for all $i \in [n']$.
Hence, the number of entries in the table is $2^{\calO(k)} \cdot n^{\calO(1)}$.
The algorithm takes  $2^{\calO(k)}\cdot n^{\calO(1)}$ time to compute each table entry.
This implies that the algorithm terminates in the desired time, and the proof of the lemma is complete.
\end{proof}

\subsubsection{Proof of \autoref{lemma:k-lesser-2k-algo}}

%\igm{I find this proof quite odd: it starts with the algorithm for the last problem we consider, instead of the first one, which I find more natural. Anyway, no problem, we can keep it like that}

%\begin{proof} (of \autoref{lemma:k-lesser-2k-algo})
\autoref{lemma:algo-const-dir-maxcut} implies that there is an algorithm, say $\calA$, that given an instance $(G, \langle X_L, X_R \rangle, k)$ of \textsc{Constrained Directed MaxCut}, runs in time $2^{\calO(k)} \cdot n^{\calO(1)}$, and correctly determines whether it is a \yes-instance.

Consider an algorithm, say $\calB$, that takes as input an instance $((G, k, d), X, \langle X_L, X_R \rangle)$ of \textsc{Constrained MaxCut} with the extra guarantee that $k = d$, and returns \yes\ or \no.
It starts by exhaustively applying \autoref{rr:only-matching}.
It then uses the reduction mentioned in \autoref{sub-sec:cont-max-cut-dir-max-cut} to obtain an instance $(D, \langle X_L, X_R \rangle, k)$ of \textsc{Constrained Directed MaxCut}.
Using Algorithm $\calA$, it determines whether the constructed instance is a \yes-instance.
If it is indeed the case, then it returns \yes, otherwise it returns \no.
This concludes the description of Algorithm~$\calB$.

The correctness of Algorithm~$\calB$ follows from  \autoref{lemma:rr-only-matching}, \autoref{lemma:redu-maxcut-digraph-maxcut}, and the correctness of Algorithm~$\calA$.
Its running time follows from the running time of Algorithm~$\calA$.
Hence, Algorithm~$\calB$ takes as input an instance $((G, k, d), X, \langle X_L, X_R \rangle)$ of \textsc{Constrained MaxCut} with a guarantee that $k  = d$, runs in time $2^{\calO(d)} \cdot n^{\calO(1)}$, and correctly determines whether it is a \yes-instance.

Algorithm~$\calB$ and \autoref{lemma:k-neq-d-to-k-eq-d} imply that there is an algorithm, say $\calC$, that given an instance $((G, k, d), X, \langle X_L, X_R \rangle)$ of \textsc{Constrained MaxCut} (without any guarantee), runs in time $2^{\calO(d)} \cdot n^{k - d + \calO(1)}$, and correctly determines whether it is a \yes-instance. Recall that since we are in the case $d \le k < 2d$, it holds that $2^{\calO(k)} = 2^{\calO(d)}$.

Consider an algorithm, say $\calD$, that takes as input an instance $((G, k, d), X, \langle X_L, X_R \rangle)$ of \textsc{Annotated Contraction(\vc)}, and returns \yes\ or \no.
It exhaustively applies \autoref{rr:anno-contr-edges-X}.
It then uses Algorithm~$\calC$ to determine whether the resulting instance is a \yes-instance of \textsc{Constrained MaxCut}.
If it is indeed the case, it returns \yes, otherwise it returns \no.
This concludes the description of Algorithm~$\calD$.

The correctness of Algorithm~$\calD$ follows from \autoref{lemma:rr-anno-contr-edges-X}, \autoref{lemma:red-const-contr-vc-maxcut}, and the correctness of Algorithm~$\calC$.
Its running time follows from the running time of Algorithm~$\calC$.
Hence, Algorithm~$\calD$ takes as input an instance $((G, k, d), X, \langle X_L, X_R \rangle)$ of \textsc{Annotated Contraction(\vc)}, runs in time $2^{\calO(d)} \cdot n^{k - d + \calO(1)}$ and correctly decides whether it is a \yes-instance.

Algorithm~$\calD$ and \autoref{lemma:contr-VC-to-annot-contr-vc} imply that there is an algorithm, say $\calE$, that given an instance $(G, k, d)$ of \textsc{Contraction(\vc)} with \autoref{guarantee:special-instance}, runs in time $2^{\calO(d)} \cdot n^{\calO(k - d)}$ and correctly decides whether it is a \yes-instance.

Consider an algorithm, say $\calF$, that takes as input an instance $(G, k, d)$ of \textsc{Contraction(\vc)} with the guarantees that $k < \rank(G)$ and $d \le k < 2d$, and returns \yes\ or \no.
It first applies the algorithm mentioned in \autoref{lemma:large-oct-large-rank-vc} as a subroutine.
If the subroutine concludes that $(G, k, d)$ is a \yes-instance, then it returns \yes.
Otherwise, let $X$ be the minimum vertex cover of $G$ returned by the subroutine.
Recall that $\rank(X) < d$.
The algorithm constructs instances $(G, k^{\circ}, d)$ for every value of $k^{\circ} \in [k]$.
In the increasing order of $k^{\circ}$, it determines whether $(G, k^{\circ}, d)$ is a \yes-instance using Algorithm~$\calE$.
If for any value of $k^{\circ}$, Algorithm~$\calE$ concludes that $(G, k^{\circ}, d)$ is a \yes-instance, then the algorithm returns \yes, otherwise it returns \no.
This concludes the description of Algorithm~$\calF$.

Note that for $k^{\circ} = 1$, instance $(G, k^{\circ}, d)$ satisfies every condition mentioned in  \autoref{guarantee:special-instance}.
The algorithm uses Algorithm~$\calE$ with $(G, k^{\circ}, d)$ as input only if Algorithm~$\calE$ concludes that $(G, k^{\circ} - 1, d)$ is a \no-instance.
The correctness and the running time of Algorithm~$\calF$ follow from those of Algorithm~$\calE$.
This implies that Algorithm~$\calF$ receives an instance $(G, k, d)$ of \textsc{Contraction(\vc)} with guarantees that $k < \rank(G)$ and $d \le k < 2d$, runs in time $2^{\calO(d)} \cdot n^{k - d + \calO(1)}$, and correctly determines whether it is a \yes-instance.
This concludes the proof of \autoref{lemma:k-lesser-2k-algo}.

\subsection{Proof of \autoref{thm:algorithm}}
\label{sub-sec:entire-algo}

%\begin{proof} (of \autoref{thm:algorithm})
Consider an instance $(G, k, d)$ of \textsc{Contraction(\vc)}.
If $G$ is a connected graph, we can use the algorithm mentioned at the start of \autoref{sec:algorithm} to conclude whether $(G, k, d)$ is a \yes-instance.
Otherwise, for each each connected component $C_i$ of $G$, integers $k_i \in [k]$, and $d_i \in [d]$, we use the algorithm to determine whether $(G[C_i], k_i, d_i)$ is a \yes-instance of \textsc{Contraction(\vc)}.
With this information, one can construct a standard Knapsack-type dynamic programming table, which is also mentioned in \cite{lima2020reducing}, to determine whether $(G, k, d)$ is a \yes-instance.

The correctness and the running time of this procedure follows from those of \autoref{lemma:k-equal-rank-algo}, \autoref{lemma:k-larger-2k-algo}, and \autoref{lemma:k-lesser-2k-algo}.
Hence, there is an algorithm that takes as input an instance $(G, k, d)$ of \textsc{Contraction(\vc)}, runs in time $2^{\calO(d)} \cdot n^{k - d + \calO(1)}$ and correctly determines whether it is a \yes-instance.
%\end{proof}

\section{Conclusion}
\label{sec:conclusion}

In this article we considered the problem of reducing the size of a minimum vertex cover of a graph $G$ by at least $d$ using at most $k$ edge contractions.
Note that the problem is trivial when $d < k$.
A few simple observations prove that when $d \le 2k$, the problem is \coNP-\Hard\ and \FPT\ when parameterized by $k + d$.
Almost all of our technical work is to handle the case when $d \le k < 2d$.
We proved that the problem is \NP-\Hard\ when $k = d + \frac{\ell - 1}{\ell + 3} \cdot d$ for any integer $\ell \ge 1$ such that $k$ is an integer (in particular, $\ell = 1$).
This implies that the problem is hard for various values of $k - d$ in the set $\{0, 1, \dots,  d - 1\}$.
We were able to prove that if $(k - d)$ is a constant then the problem is \FPT\ when parameterized by $k + d$.
However, if no such a condition is imposed, then the problem is \W[1]-\Hard.
More precisely, we presented an algorithm with running time $2^{d} \cdot n^{k - d + \calO(1)}$ and proved that the problem is \W[1]-\Hard\ when parameterized by $k + d$ in the case where $k - d = \frac{d}{3}$ (see the proof of \autoref{thm:w1-hard}).

We believe that it should be possible to prove that the problem is \NP-\Hard\ for \emph{every} value of $k - d$ in the set $\{0, 1, \dots,  d - 1\}$.
Such a reduction has the potential to sharpen the distinction between \FPT\ and \W[1]-\Hard\ cases as $k - d$ varies in this range.
It might also simplify the analysis of our \XP\ algorithm or lead to a simpler algorithm.
It would be interesting to analyze the parameterized complexity of the problem with respect to structural parameters like the vertex cover number or the treewidth of the input graph.
Note that the problem is trivially \FPT\ when parameterized by the vertex cover number. Finally, it is worth mentioning that we did not focus on optimizing the degree of the polynomial term $n^{\calO(1)}$ in our \XP\ algorithm, although it is reasonably small.

%To the best of our knowledge, this is the only blocker problem that admits an \XP\ algorithm when parameterized by the budget $k$ and demand $d$ when the modification operation is edge contraction.
%We believe understand this problem from the lenses of exact exponential algorithms or lossy kernelization will enhance our understanding of blocker problems when the modification operation is edge contraction.
%\ig{If there is time, we should unify the format of the references. I would always use full names of journals, add the DOI to all references, and put a uniform format for the conferences, such as ``Proc. of the 10th... Computation (IPEC), volume XXX of LIPIcs, pages YYYY, ...''}
%%
%% Bibliography
%%

%% Please use bibtex,

\bibliography{references}

\end{document}